\documentclass[fleqn,envcountsame]{LMCS}

\def\doi{8(3:21)2012}
\lmcsheading%
{\doi}
{1--24}
{}
{}
{Dec.~23, 2011}
{Sep.~20, 2012}
{}
 
\usepackage{tikz}
\usetikzlibrary{shapes}
\usetikzlibrary{snakes}
\usetikzlibrary{arrows} 
\usepackage{mathtools}
\usepackage{stmaryrd}
\usepackage{xspace}
\usepackage{enumerate,hyperref}
\usepackage{amssymb,amsmath}
\usepackage{url}

\makeatletter
\newcommand{\@abbrev}[3]{
  \def\c@a@def##1{
      \if ##1.
        \relax
      \else
        \@ifdefinable{\@nameuse{#1##1}}{\@namedef{#1##1}{#2##1}}
        \expandafter\c@a@def
      \fi
    }
  \c@a@def #3.
}
\@abbrev{bb}{\mathbb}{ABCDEFGHIJKLMNOPQRSTUVWXYZ}
\@abbrev{bf}{\mathbf}{ABCDEFGHIJKLMNOPQRSTUVWXYZabcdefghijklmnopqrstuvwxyz}
\@abbrev{bit}{\boldsymbol}{ABCDEFGHIJKLMNOPQRSTUVWXYZabcdefghijklmnopqrstuvwxyz}
\@abbrev{cal}{\mathcal}{ABCDEFGHIJKLMNOPQRSTUVWXYZ}
\@abbrev{frak}{\mathfrak}{ABCDEFGHIJKLMNOPQRSTUVWXYZabcdefghijklmnopqrstuvwxyz}
\@abbrev{rm}{\mathrm}{ABCDEFGHIJKLMNOPQRSTUVWXYZabcdefghijklmnopqrstuvwxyz}
\@abbrev{scr}{\mathscr}{ABCDEFGHIJKLMNOPQRSTUVWXYZ1234567890}
\@abbrev{sf}{\mathsf}{ABCDEFGHIJKLMNOPQRSTUVWXYZabcdefghijklmnopqrstuvwxyz}
\makeatother

\newcommand{\val}{\ensuremath{\mathrm{val}}}

\newcommand{\cf}{cf.\xspace}

\newcommand{\eg}{e.g.\xspace}
\newcommand{\ie}{i.e.\xspace}

\newcommand{\muc}{$\mu$-calculus\xspace}
\newcommand{\qmu}{quantitative \muc}

\newcommand{\Qmu}{\ensuremath{Q\mu}\xspace} 

\newcommand{\pzero}{Player~$0$\xspace}
\newcommand{\pone}{Player~$1$\xspace}

\newcommand{\Rinf}{\ensuremath{\mathbb{R}_{\infty}}}
\newcommand{\Qinf}{\ensuremath{\mathbb{Q}_{\infty}}}
\newcommand{\Zinf}{\ensuremath{\mathbb{Z}_{\infty}}}

\newcommand{\setR}{\ensuremath{\mathbb{R}}}
\newcommand{\setQ}{\ensuremath{\mathbb{Q}}}
\newcommand{\setZ}{\ensuremath{\mathbb{Z}}}

\newcommand{\Pot}{\ensuremath{\calP}}
\newcommand{\Func}{\ensuremath{\mathcal{F}}}
\newcommand{\Diam}{\ensuremath{\Diamond}}
\newcommand{\ol}[1]{\ensuremath{\overline{#1}}}
\newcommand{\mc}{\ensuremath{\mathrm{MC}}} 
\newcommand{\sem}[1]{\ensuremath{\llbracket#1\rrbracket}}
\newcommand{\semk}[1]{\ensuremath{\llbracket#1\rrbracket^{\calK}}}
\newcommand{\lint}{\ensuremath{\frakI}} 
\newcommand{\semek}[1]{
  \ensuremath{\llbracket#1\rrbracket_{\lint}^{\calK}}
}
\newcommand{\semsub}[1]{
  \ensuremath{\llbracket#1\rrbracket_{\lint[X \leftarrow f]}^{\calK}}
}

\newcommand{\floor}[1]{\ensuremath{\lfloor #1 \rfloor}}
\newcommand{\ceil}[1]{\ensuremath{\lceil #1 \rceil}}
\newcommand{\intervals}{\ensuremath{\calI(\Rinf)}}
\newcommand{\intervalsplus}{\ensuremath{\calI(\Rinf^{\geq 0})}}
\newcommand{\labels}{\ensuremath{\lambda}}
\newcommand{\successor}{\ensuremath{\mathrm{succ}}}
\newcommand{\coeff}{\ensuremath{\delta}} 
\newcommand{\indexi}{\ensuremath{\iota}}
\newcommand{\strat}{\ensuremath{\sigma}} 
\newcommand{\pay}{\ensuremath{\mathrm{p}}} 
\newcommand{\di}{\ensuremath{\mathrm{d}}} 
\newcommand{\dih}{\ensuremath{\mathrm{d^*}}}
\newcommand{\corr}{\ensuremath{\mathrm{c}}}
\newcommand{\dleft}{\ensuremath{\mathrm{d}_l}}
\newcommand{\dright}{\ensuremath{\mathrm{d}_r}}
\newcommand{\loc}{\ensuremath{\mathrm{loc}}} 
\newcommand{\var}{\ensuremath{\mathrm{var}}} 
\renewcommand\epsilon{\varepsilon}
\renewcommand\phi{\varphi}
\newcommand{\nin}{\ensuremath{\not \in}}
\renewcommand{\max}{\mathrm{max}}
\newcommand{\cupdot}{\mathbin{\dot{\cup}}}
\newcommand{\closedots}{\mathinner{\!\ldotp\!\ldotp\!\ldotp\!}}

\begin{document}

\title[Model Checking the Quantitative $\mu$-Calculus on
  Linear Hybrid Systems]{Model Checking the Quantitative $\mu$-Calculus on
  Linear Hybrid Systems}
\author[D.~Fischer]{Diana~Fischer\rsuper a}
\address{{\lsuper a}Mathematische Grundlagen der Informatik, RWTH Aachen University}
\email{fischer@logic.rwth-aachen.de}

\author[\L.~~Kaiser]{\L{}ukasz~Kaiser\rsuper b}
\address{{\lsuper b}LIAFA, CNRS \& Universit{\'e} Paris Diderot -- Paris 7}
\email{kaiser@liafa.univ-paris-diderot.fr}

\thanks{Authors were supported by DFG~AlgoSyn~1298 and
  ANR~2010~BLAN~0202~02~FREC}


\begin{abstract}
We study the model-checking problem for a quantitative
extension of the modal $\mu$-calculus on a class of hybrid systems.
Qualitative model checking has been proved decidable and implemented for
several classes of systems, but this is not the case for quantitative
questions that arise naturally in this context. Recently, quantitative
formalisms that subsume classical temporal logics and allow
the measurement of interesting quantitative phenomena were introduced.
We show how a powerful quantitative logic, the quantitative $\mu$-calculus, can
be model checked with arbitrary precision on initialised linear hybrid systems.
To this end, we develop new techniques for the discretisation of continuous
state spaces based on a special class of strategies in model-checking games
and present a reduction to a class of counter parity games.
\end{abstract}

\maketitle

\section{Introduction}

Modelling discrete-continuous systems by a hybrid of a discrete transition
system and continuous variables which evolve according to a set of
differential equations is widely accepted in engineering.
While model-checking techniques have been applied to verify safety, liveness
and other temporal properties of such systems \cite{ACHHHNOSY95,HHM99,HKPV95},
it is also interesting to infer quantitative values for certain queries.
For example, one may not only want to check that a variable of a system
does not exceed a given threshold, but also to compute the maximum value
of the variable over all runs, checking whether any such threshold exists.

Thus far, quantitative testing of hybrid systems has only been done
by simulation, and hence lacks the strong guarantees which can be given
by model checking. In recent years, there has been a strong interest
in extending classical model-checking techniques and logics to the quantitative
setting. Several quantitative temporal logics have been introduced,
see \eg  \cite{Alfaro03,AlfaroFS04,AlfaroM04,FGK10,seidl07,Hugo07,McIver},
together with model-checking algorithms for simple classes of systems,
such as finite transition systems with discounts. Still, none of those
systems allowed for dynamically changing continuous variables.
We present the first model-checking algorithm for a non-stochastic quantitative
temporal logic on a class of hybrid systems.
The logic we consider, the quantitative $\mu$-calculus \cite{FGK10},
is based on a formalism first introduced in \cite{AlfaroFS04}.
It properly subsumes the standard $\mu$-calculus, \cf \cite{BradfieldS01}, and
thus also CTL and LTL.
Therefore the present result, namely that it is possible to
model check quantitative $\mu$-calculus on initialised linear
hybrid systems, properly generalises a previous result on
model checking LTL on such systems \cite{HHM99,HKPV95}, which is
one of the strongest model-checking results for hybrid systems.

The restriction to initialised linear systems is made because verification
of temporal properties over general hybrid systems is undecidable. This
holds even for linear systems, thus one must pick an appropriate abstraction
of the system. An established and very well-studied way to do this is to
first approximate the continuous behaviour of the variables by linear
behaviour in a finite number of intervals.
This method, applied to a number of functions $f_1(x), \ldots, f_m(x)$ that
evolve according to a set of arbitrary differential equations
$\calD(f_1, \ldots, f_m)$, generates a set of disjoint intervals
$I_1, \ldots, I_k$ with $I_1 \cup \ldots \cup I_k = \setR$ and
a set of linear coefficients $a_i^j, b_i^j$ such that in $I_j$ it is
approximately true that $f_i(x) = a_i^j \cdot x + b_i^j$, \ie the
derivative $\frac {df_i}{dt} = a_i^j$. There are several
ways to generate such linear approximations of solutions of differential
equations and, depending on the method in question, one can obtain various
kinds of error bounds for the respective classes of functions.
We do not investigate these issues (or other approximation methods) here,
but focus instead on the linear system obtained. 

As stated above, even simple
qualitative verification problems are undecidable for general hybrid systems.
This remains true even after the natural approximation by a linear system.
Hence, one more assumption is made, namely that if the speed of evolution of a
variable changes between discrete locations then also the variable is reset
on that transition. Systems with this property, called \emph{initialised}
linear systems, are -- besides o-minimal systems \cite{LPS00,BBC07} and their
recent extensions \cite{VPVD08} -- one of the largest classes of hybrid systems
with decidable temporal logic \cite{HKPV95}.
Observe that when an arbitrary hybrid system is approximated by a linear one,
one can \emph{try} to directly obtain an initialised system by computing
boundary values \cite{HHW96}. This can be done by either assuring that discrete
transitions are taken only at the borders of the intervals $I_j$, or by taking
a finer subdivision of the intervals to increase the precision of coordination
between the discrete and the continuous part of the system. Note that, even
though this procedure has been implemented in model-checking programs, it is
only a heuristic -- it necessarily fails for general systems for which the
model-checking problem is undecidable.

The logic we study is quantitative -- it allows to express properties involving
suprema and infima of values of the considered variables during runs that
satisfy various temporal properties, \eg to answer
``what is the maximal temperature on a run during which a safety condition
holds?''. To model check formulae of the quantitative $\mu$-calculus,
we follow the classical parity game-based approach
and adapt some of the methods developed in the qualitative
case and for timed systems. To our surprise, these methods turned out
not to be sufficient and did not easily generalise to
the quantitative case. As we will show below, the quantitative systems we study
behave in a substantially different way than their qualitative counterparts.
We overcome this problem by working directly with a quantitative
equivalence relation, roughly similar to the region graph for timed automata,
and finally by exploiting a recent result on counter parity games.

\textbf{Organisation.}
The organisation of this paper follows the reductions needed to model check
a formula $\phi$ over a hybrid system $\calK$. 
In Section \ref{sec_hybrid}, we introduce the necessary notation,
the systems and the logic. Then, we present
an appropriate game model in Section \ref{sec_ipg} and show how to 
construct a model-checking game $\calG$ for the system and the formula.
In Section \ref{sec_basic}, we transform the interval games constructed for
arbitrary initialised linear hybrid systems to flat games, where the linear
coefficients are always $1$. In Section \ref{sec_ds}, we show how the strategies
can be discretised and still lead to a good approximation of the original game.
Finally, in Section \ref{sec_pushdown}, we reduce the problem to counter
parity games and exploit a recent result to solve them.
To sum up, the steps taken are depicted below.
\[ \calK,\phi \leadsto \text{ model-checking game } \calG \leadsto
   \text{ flat } \calG \leadsto
   \text{ counter-reset } \calG \leadsto \text{value.} \]

\section{Hybrid Systems and Quantitative Logics} \label{sec_hybrid}
We denote the real and rational numbers and integers extended with both
$\infty$ and $-\infty$ by $\Rinf$, $\Qinf$ and $\Zinf$ respectively.
We write $\calI(\Zinf), \calI(\Qinf)$ and $\intervals$ for all open or closed
intervals over $\Rinf$ with endpoints in $\Zinf, \Qinf$ and $\Rinf$.

\begin{defi}
A \emph{linear hybrid system over $M$ variables},
$\calK = (V, E, \{P_i\}_{i \in J}, \lambda, \coeff)$,
is based on a directed graph $(V,E)$, consisting of a set of locations $V$
and transitions $E \subseteq V \times V$. The labelling function
$\lambda : E \to \Pot_{\mathrm{fin}}(\calL_M)$ assigns to each transition
a finite set of labels. The set $\calL_M$ of \emph{transition labels}
consists of triples $l = (I, \ol{C}, R)$, where the vector
$\ol{C} = (C_1, \ldots, C_M)$ (with $C_i \in \intervals$ for
$i \in \{1, \ldots, M\}$) represents the constraints each of the variables
needs to satisfy for the transition to be allowed, the interval
$I \in \intervalsplus$ represents the possible period of time
that elapses before the transition is taken,
and the reset set $R$ contains the indices of the
variables that are reset during the transition, 
\ie $i \in R$ means that $y_i$ is set to zero.
For each $i$ of the finite index set $J$,
the function $P_i: V \to \Rinf$ assigns to each location
the value of the static quantitative predicate $P_i$.
The function $\coeff\ :\ V \to \setR^{M}$ assigns to each location
and variable $y_i$ the coefficient $a_i$ such that the variable evolves
in this location according to the equation $\frac {dy_i}{dt} = a_i$.
\end{defi}

Please note that although we do not explicitly have any invariants
(or constraints) in locations, we can simulate them by choosing either
the time intervals or variable constraints on the outgoing transitions
accordingly. If the values of predicates and labels range over $\Qinf$
or $\Zinf$ instead of $\Rinf$, we talk about linear hybrid systems
over $\setQ$ and $\setZ$, respectively. 

The \emph{state} of a linear hybrid system $\calK$ is a location
combined with a valuation of all $M$ variables, $S = V \times \Rinf^M$.
For a state $s = (v, y_1, \ldots, y_M)$ we say that a transition
$(v, v') \in E$ is \emph{allowed} by a label $(I, \ol{C}, R) \in \labels((v, v'))$
if $\ol{y} \in \ol{C}$ (\ie if $y_i \in C_i$ for all $i = 1, \ldots, M$).
We say that a state $s' = (v', y_1', \ldots, y_M')$ is a successor of $s$,
denoted $s' \in \successor(s)$, when there is a transition $(v,v') \in E$,
allowed by label $(I, \ol{C}, R)$, such that $y'_i = 0$ for all $i \in R$ and
there is a $t \in I$ such that $y'_i = y_i + (a_i \cdot t)$ where 
$a_i=\coeff_i(v)$ for all $i \not \in R \in \labels((v,v'))$.
A run of a linear hybrid system starting from location $v_0$ is
a sequence of states $s_0, s_1, \ldots$ such that $s_0 = (v_0, 0, \ldots, 0)$
and $s_{i+1} \in \successor(s_i)$ for all $i$. Given two states $s$ and
$s' \in \successor(s)$ and a reset set $R \neq \{1,\ldots,M\}$ we denote by
$s' -_R s$ the increase of the non-reset variables that occurred during the
transition, \ie $\frac {y'_i - y_i} {a_i}$ for some $i \not\in R$ where
$s=(v,\ol{y})$ and $s'=(v',\ol{y}')$. 

\begin{defi}
A linear hybrid system $\calK$ is \emph{initialised} if for each
$(v,w) \in E$ and each variable $y_i$ it holds that if 
$\coeff_i(v) \neq \coeff_i(w)$ then $i \in R$ for $R \in\lambda((v,w))$.
\end{defi}

Intuitively, an initialised system cannot store the value of a variable
whose evolution rate changes from one location to another.

\begin{exa}
To clarify the notions we use, we consider a variant of a standard example 
for a linear hybrid system, the leaking gas burner.

\begin{figure}
\begin{center}
\begin{tikzpicture}
\node[draw,circle] (a) at (-2, 0) {$v_0$};
\node[] () at (-2, 0.75) {{\small{$P=\infty$}}};
\node[] () at (-2.5, -0.75) {\emph{\small{$\frac{dy_0}{dt}=1,\frac{dy_1}{dt}=1$}}};
\node[draw, circle] (b) at (2, 0) {$v_1$};
\node[] () at (2, 0.75) {{\small{$P=-\infty$}}};

\node[] () at (2.5, -0.75) {\emph{\small{$\frac{dy_0}{dt}=1,\frac{dy_1}{dt}=0$}}};
\path (a) edge[bend right,->] node[above] {{\tiny{$[0, 1]$}}}
node[below] {{\tiny{$R=\{y_0\}$}}} (b);
\path (b) edge[bend right,->] node[below] {{\tiny{$[0,\infty)$}}}
  node[above] {{\tiny{$R=\{y_0\}$,$y_0 \in [30,40]$}}} (a);
\end{tikzpicture} 
\end{center}
\caption{Leaking gas burner LHS $\calL = (V, E, P, \lambda,\coeff)$ (not initialised)}\label{leaking-not-in}
\end{figure}
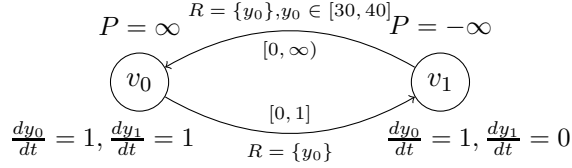

Our version is depicted in Figure \ref{leaking-not-in}. 
This system represents a gas valve that can leak gas to a burner, so it has
two states: $v_0$, where the valve is open (and leaking gas) and $v_1$ where
it is closed. This is also indicated by a qualitative predicate $P$ that
has the value $\infty$ if the gas is leaking (in location $v_0$) and
$-\infty$ otherwise.
The system has two variables. The first variable, $y_0$, is a clock
measuring the time spent in each location, and is reset on each transition,
\ie after each discrete system change. The variable $y_1$ is a stop watch
and measures the total time spent in the leaking location.
Thus, this system is not initialised.
The time intervals on the transitions control the behaviour of the system.
On the transition $(v_0,v_1)$ there are no restrictions on the variables, but
we are only allowed to choose a time unit from $[0,1]$, \ie we can stay a
maximum of one time unit in location $v_0$. On the transition $(v_1,v_0)$ there 
is a restriction on the value of $y_0$, it has to have a value between 30
and 40 for this transition to be allowed, while there is no restriction on
the choice for the time unit (of course, this could also be modelled the other
way around).
Intuitively, the time intervals indicate that the gas valve will
leak gas for a time interval between 0 and 1 seconds and 
then be stopped and that it can only leak again after
at least 30 time units.\\

In Figure \ref{leaking}, we show an initialised version of the leaking gas burner.
The only difference is that $y_1$ is not a stop watch anymore but a normal
clock. Since now both variables are just clocks (which means that their
evolution rates are one everywhere), the system is trivially initialised.

\begin{figure}
\begin{center}
\begin{tikzpicture}
\node[draw,circle] (a) at (-2, 0) {$v_0$};
\node[] () at (-2, 0.75) {{\small{$P=\infty$}}};
\node[] () at (-2.5, -0.75) {\emph{\small{$\frac{dy_0}{dt}=1,\frac{dy_1}{dt}=1$}}};
\node[draw, circle] (b) at (2, 0) {$v_1$};
\node[] () at (2, 0.75) {{\small{$P=-\infty$}}};
\node[] () at (2.5, -0.75) {\emph{\small{$\frac{dy_0}{dt}=1,\frac{dy_1}{dt}=1$}}};
\path (a) edge[bend right,->] node[above] {{\tiny{$[0, 1]$}}}
node[below] {{\tiny{$R=\{y_0\}$}}} (b);
\path (b) edge[bend right,->] node[below] {{\tiny{$[0,\infty)$}}}
  node[above] {{\tiny{$R=\{y_0\}$,$y_0 \in [30,40]$}}} (a);

\end{tikzpicture} 
\end{center}
\caption{Leaking gas burner LHS $\calL = (V, E, P, \lambda,\coeff)$ (initialised)}\label{leaking}
\end{figure}
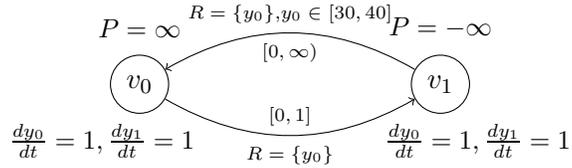
\end{exa}

\subsection{Quantitative $\mu$-Calculus}
In this section, we present a version of the quantitative $\mu$-calculus first
introduced in \cite{FGK10}.
The version we use here is additive and includes variables. It is evaluated on
linear hybrid systems.

\begin{defi}
Given sets of fixpoint variables $\calX$, system variables $\{y_1, \ldots, y_M\}$ and predicates $\{P_i\}_{i \in J}$,
the formulae of the \emph{quantitative \muc (\Qmu) with variables}
are given by the EBNF grammar:
\[ \phi \ ::= \ P_i \mid X_j \mid y_k \mid 
   \neg \phi \mid \phi \land \phi \mid \phi \lor \phi \mid
   \Box \phi \mid \Diamond \phi 
   \mid \mu X_j.\phi \mid \nu X_j. \phi \, , \]
where $X_j \in \calX, y_k \in \{y_1, \ldots, y_M\}$, and in the cases
$\mu X_j.\phi$ and $\nu X_j. \phi$, the variable $X_j$ must appear positively
in $\phi$, \ie under an even number of negations.
\end{defi}

Let $\Func = \{f\ :\ S \to \Rinf\}$.
Given an interpretation $\lint : {\calX} \to \Func$,
a variable $X \in \calX$, and a function $f \in \Func$, we denote by
$\lint[X \leftarrow f]$ the interpretation $\lint'$,
such that $\lint'(X) = f$ and $\lint'(X') = \lint(X')$
for all $X' \neq X$.

\begin{defi}
Given a linear hybrid system
$\calK = (V, E, \lambda, \{P_i\}_{i \in J}, \coeff)$ and an interpretation
$\lint$, a $\Qmu$-formula yields a valuation
function $\semek{\phi} : S \to \Rinf$ defined in the following standard
way for a state $s = (v^s, y^s_1, \ldots, y^s_M)$.
\begin{iteMize}{$\bullet$}
\item $ \semek{P_i}(s) = P_i(v^s)$, 
      $ \semek{X}(s) = \lint(X)(s)$, and
      $ \semek{y_i}(s) = y^s_i$, $ \semek{\neg \phi} = - \semek{\phi}$
\item $ \semek{\phi_{1} \land \phi_{2}} =
        \min \{\semek{\phi_{1}},\semek{\phi_{2}}\}$ and
      $ \semek{\phi_{1} \lor \phi_{2}} =
        \max \{\semek{\phi_{1}},\semk{\phi_{2}}\}$,
\item $ \semek{\Diam \phi}(s) =
        \sup_{s' \in \successor(s)} \semek{\phi}(s')$ and
      $ \semek{\Box \phi}(s) =
        \inf_{s' \in \successor(s)} \semek{\phi}(s')$,
\item $\semek{\mu X.\phi} = \inf \{ f \in \Func : f = \semsub{\phi}\}$,\\
  $\semek{\nu X.\phi} = \sup \{ f \in \Func : f = \semsub{\phi}\}$.
\end{iteMize}
\end{defi}

\noindent For formulae without free variables we write $\semk{\phi}$
rather than $\semek{\phi}$.

Please note that the inclusion of variables does not fundamentally 
change the semantics of \qmu. The \qmu in \cite{FGK10} is evaluated
on quantitative transition systems. Here, a formula is evaluated on
the state graph of a linear hybrid system, rather than the system itself.
Intuitively, a linear hybrid system is a compact representation
of an infinite quantitative transition system (its state graph).
Thus, many properties of the \qmu from \cite{FGK10} remain true.
For example, to embed the classical $\mu$-calculus in \qmu one must
interpret \emph{true} as $+\infty$ and \emph{false} as $-\infty$.

\begin{exa}
The formula $\mu X.(\Diam X \lor y_1)$ evaluates to the supremum of
the values of $y_1$ on all runs from some initial state:
e.g. to $\infty$ if evaluated on the simple initialised
leaking gas burner model. To determine the longest period of time
during which the gas is leaking we use the formula
$\mu X.(\Diam X \lor (y_0 \land P))$, which evaluates to $1$ on
the initial state $(v_0, \ol{0})$ in our example.
\end{exa}

The remainder of this paper is dedicated to the proof of our following
main result which shows that $\semk{\phi}$ can be approximated
with arbitrary precision on initialised linear hybrid systems.

\begin{thm} \label{mainthm}
Given an initialised linear hybrid system $\calK$, a \qmu formula $\phi$
and an integer $n > 0$, it is decidable whether $\semk{\phi} = \infty$,
$\semk{\phi} = -\infty$, or else a number $r \in \bbQ$ can be computed
such that $|\semk{\phi} - r| < \frac 1 n$.
\end{thm}

In other words, for every $\epsilon$ we can approximate $\semk{\phi}$
within $\epsilon$. We formulated the theorem above using $n$ because it
makes the representation of $\epsilon$ precise, so we can provide a complexity
bound: Given on input the system $\calK$, the formula $\phi$ and $n$, we will
show how to compute the number $r$ (or output $\pm\infty$) in \textsc{8EXPTIME}.

\section{Interval Games}\label{sec_ipg}

In this section, we define a variant of quantitative parity games
suited for model checking \Qmu on linear hybrid systems.
As mentioned above, a linear hybrid system can be seen as a compact 
representation of an infinite quantitative transition system. 
Similarly, a parity game that is played on a linear hybrid
system can be viewed as a compact, finite description of an
infinite quantitative parity game, as defined in \cite{FGK10}.

\begin{defi} \label{def-ipg}
An \emph{interval parity game} (IPG)
$\calG = (V_0, V_1, E, \labels, \coeff, \indexi, \Omega)$,
is played on a LHS $(V, E, \labels, \coeff)$ (without predicates) and
$V = V_0 \cupdot V_1$ is divided into positions of either \pzero or 1. 
The transition relation  $E \subseteq V \times V$ describes possible
moves in the game which are labelled by the function
$\labels : E \to \Pot_{\mathrm{fin}}(\calL_M)$.
The function $\indexi : V \to M \times \Rinf \times \Rinf$ assigns to each
position the index of a variable and a multiplicative and additive factor,
which are used to calculate the payoff if a play ends in this position.
The priority function $\Omega : V \to \{0, \ldots, d\}$
assigns a priority to every position.
\end{defi}

Please note that interval parity games are played on linear hybrid systems without any quantitative predicates,
\ie the set of of predicates is empty and therefore omitted.

A \emph{state} $s=(v, \ol{y}) \in V \times \Rinf^M$ of an interval game
is a position in the game graph together with a variable assignment for
all $M$ variables. 
A state $s'$ is a successor of $s$ if it is a successor
in the underlying LHS, \ie if $s' \in \successor(s)$.
We use the functions $\loc(s)=v$ and $\var(s)=\ol{y}, \var_i(s)=y_i$ to 
access the components of a state.
For a real number $r$, we denote by $r \cdot s=(v, r \cdot \var_0(s), \ldots 
r \cdot \var_M(s))$ and  $r + s=(v, r + \var_0(s), \ldots r + \var_M(s)).$
We call $S_i$ the state set $\{s=(v, \ol{y}) : v \in V_i \}$ where player $i$
has to move and $S = S_0 \cupdot S_1$.

\textbf{How to play.}
Every play starts at some position $v \in V$ with all variables set to $0$,
\ie the starting state is $s_0=(v, 0, \ldots, 0)$.
For every state $s = (v, \ol{y}) \in S_i$, player $i$
chooses an allowed successor state $s' \in \successor(s)$ and the play
proceeds from $s'$. If the play reaches a state $s$ such that
$\successor(s) = \emptyset$ it ends, otherwise the play is infinite.

Intuitively, the players choose the time period they want to spend
in a location before taking a specified transition.
Note that in this game every position could possibly be a terminal
position. This is the case if it is not possible to choose a time period
from the given intervals in such a way that the respective constraints
on all variables are fulfilled.

\textbf{Payoffs.}
The outcome $\pay(s_0 \closedots s_k)$ of a finite play ending in 
$s_k=(v, y_1, \closedots, y_M)$ where $\indexi(v) = (i, a, b)$ is
$\pay(s_k)= a \cdot y_i + b$.
To improve readability, from now on we will simply write 
$\indexi(v) = a \cdot y_i + b$ in this case.
The outcome of an infinite play depends only on the lowest priority seen
infinitely often in positions of the play. We will assign the value
$- \infty$ to every infinite play, where the lowest priority seen infinitely
often is odd, and $\infty$ to those where it is even.

\textbf{Goals.}
The two players have opposing objectives regarding the outcome of the play.
\pzero wants to maximise the outcome, while \pone wants to minimise~it.

\textbf{Strategies.} 
A strategy for player $i \in {0,1}$ is a function
$\strat : S^*S_i \to S$ with $\strat(s) \in \successor(s)$.
A play $\pi = s_0 s_1 \ldots$ is \emph{consistent with a strategy} $\strat$
for player $i$, if $s_{n+1}=\strat(s_0\ldots s_n)$ for every $n$ such that
$s_n\in S_i$.
For strategies $\sigma, \rho$ for the two players, 
we denote by $\pi(\sigma, \rho,s)$ the unique
play starting in state $s$ which is consistent with both $\sigma$ and
$\rho$.

\textbf{Determinacy.} 
A game is \emph{determined} if, for each state $s$, the highest outcome
\pzero can assure from this state and the lowest outcome \pone can assure
coincide,
\[ \adjustlimits\sup_{\sigma \in \Gamma_0} \inf_{\rho \in \Gamma_1}
      \pay(\pi(\sigma, \rho,s)) = 
  \adjustlimits\inf_{\rho \in \Gamma_1} \sup_{\sigma \in \Gamma_0}
      \pay(\pi(\sigma, \rho,s)) =: \val \calG (s),\]
where $\Gamma_0, \Gamma_1$ are the sets of all possible strategies 
for \pzero, \pone and the achieved outcome is called the 
\emph{value of $\calG$ at $s$}.

We say that the interval game is over $\setQ$ or $\setZ$ if both
the underlying LHS and all constants in $\iota(v)$
are of the respective kind. Please note that this does
not mean that the players have to choose their values from $\setQ$ or
$\setZ$, just that the endpoints of the intervals and constants
in the payoffs are in those sets.

Intuitively, in a play of an interval parity game, the players choose
successors of the current state as long as possible.
\begin{exa}
In Figure \ref{ipg}, we show a simple example of an interval parity game.
Positions of \pzero are depicted as circles and positions of \pone
as boxes. To keep things simple, there is just one clock variable, 
$y_0$, all constraints are trivially true and the reset
sets are empty, so we label the transitions only with the time intervals
that the players can choose from. 
The priorities are depicted next to the nodes for non-terminal
positions and the evaluation function above the terminal position (in general, also
positions with outgoing edges could be terminal, however in this example this is
not possible as there are no constraints on the variable).

A play of this system starting at node $v_0$ could end after two
moves in position $v_2$, if \pone decided to move there (he also has the choice to
move down). 
The payoff of this play would then
depend only on the choice that \pzero made in the first move, for example $\frac 1 3 \in [0, \frac 1 2]$. 
Then the payoff would be
$3 \cdot (\frac 1 3 + 2) - 1=6$ (as in this 
play, the second time interval only permits the choice $2$).

If \pone would move down instead
of ending the play and the play would loop infinitely often in the cycle $v_3, v_4, v_5$
at the bottom, the least priority that occurs infinitely often would 
determine the outcome of the play; in this case it would be 0 at $v_3$ and therefore
the payoff would be $\infty$.
\end{exa}
\begin{figure}[h]\label{ipg}
\begin{center}
\begin{tikzpicture}

\node [circle,draw] (a) at (0,1) {$v_0$};
\node [circle] (a1) at (1.3,1) {\small{$\Omega(v_0)=1$}};

\node [rectangle,draw,minimum size =0.6cm] (b) at (0,-0.5) {$v_1$};
\node [rectangle] (b1) at (1.3,-0.5) {\small{$\Omega(v_1)=1$}};

\node [rectangle,draw, minimum size =0.6cm] (c) at (-2,-0.5) 
  {\small{{$v_2$}}};
\node [rectangle] (c1) at (-2.5,0.1) 
  {\small{{$\iota(v_2)=3 \cdot y_0 - 1$}}};

\node [circle,draw] (d) at (0,-2) {$v_3$};
\node [circle] (d1) at (-1.3,-2) {\small{$\Omega(v_3)=0$}};

\node [circle,draw] (e) at (-2,-3) {$v_4$};
\node [circle] (e1) at (-3.3,-3) {\small{$\Omega(v_4)=2$}};

\node [circle,draw] (f) at (2,-3) {$v_5$};
\node [circle] (f1) at (3.3,-3) {\small{$\Omega(v_5)=1$}};

\path[->] (a) edge []node[right] {{\tiny{$[0,\frac 1 2]$}}} (b);
\path[->] (b) edge []node[above] {{\tiny{$[2,2]$}}} (c);

\path[->] (b) edge []node[right] {{\tiny{$[1,1]$}}} (d);
\path[->] (d) edge []node[above] {{\tiny{$[1,1]$}}} (e);
\path[->] (f) edge []node[above] {{\tiny{$[1,1]$}}} (d);
\path[->] (e) edge []node[above] {{\tiny{$[1,1]$}}} (f);
\path[->] (f.north) edge [bend right]node[right] {{\tiny{$[1,1]$}}} (b.east);
\end{tikzpicture}
\end{center}
\caption{Simple interval parity game} 
\end{figure}
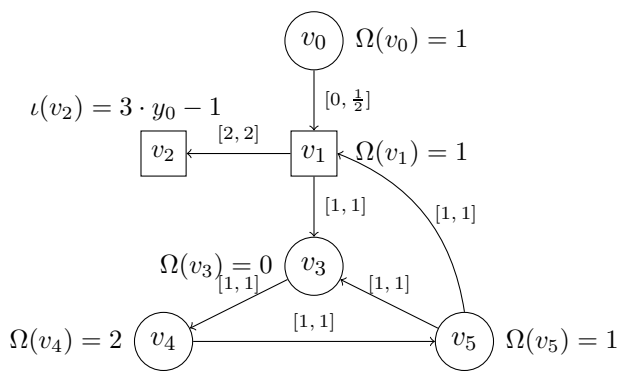

We already mentioned that an interval parity game can be seen
as a representation of a quantitative parity game, now we 
want to describe this formally.
We use the notion from \cite{FGK10} and define,
for an IPG with $M$ variables
$\calG = (V_0, V_1, E, \labels, \coeff, \indexi, \Omega)$,
the corresponding infinite quantitative parity game without discounts
$\calG^* = (V_0 \times \Rinf^M, V_1 \times \Rinf^M, E^*, \lambda^*, \Omega^*)$
with $(s, s') \in E^*$ iff $s'$ is a successor of $s$ as above,
$\Omega^*(v, \ol{z}) = \Omega(v)$ and
$\lambda^*(v, \ol{z})=\alpha \cdot z_i + \beta $ iff
$\indexi(v)=\alpha \cdot y_i + \beta$.
The notions of plays, strategies, values and determinacy 
for the IPG $\calG$ are defined exactly as the ones for
the quantitative parity game $\calG^*$ in \cite{FGK10}.
In particular, it follows from the determinacy of
quantitative parity games that also interval parity games
are determined.

\subsection{Model-Checking Games for \Qmu} \label{subsec_mc}

A game $(\calG, v)$ is a model-checking game for a formula $\phi$ and
a system $\calK, v'$, if the value of the game starting from $v$ is
exactly the value of the formula evaluated on $\calK$ at $v'$.
In the qualitative case, that means, that $\phi$ holds in $\calK, v'$ if
\pzero wins in $\calG$ from~$v$. For a linear hybrid system $\calK$
and a \Qmu-formula $\phi$, we construct an IPG
$\mc[\calK,\phi]$ which is the model-checking game for $\phi$ on $\calK$.

The full definition of $\mc[\calK,\phi]$ closely follows the construction
presented in~\cite{FGK10} and is presented below.

Intuitively, the positions are pairs consisting of a subformula
of $\phi$ and a location of $\calK$. 
Which player moves at which position depends on the outermost 
operator of the subformula. At disjunctions \pzero moves
to a position corresponding to one of the disjuncts and from 
$(\Diamond \phi, v)$ to $(\phi, w)$ where 
$(v,w) \in E^\calK$, and \pone makes analogous
moves for conjunctions and $\Box$. From fixed-point variables the play
moves back to the defining formula and the priorities
of positions depends on the alternation level of fixed points, 
assigning odd priorities to least fixed points and even priorities 
to greatest fixed points.

\begin{defi}
For a linear hybrid system $\calK = (V, E, \{P_i\}_{i \in J}, \labels, \coeff)$
and a \Qmu-formula $\phi$ in negation normal form, the interval game 
\[ \mc[\calK,\phi] = (V_0, V_1, E, \labels, \coeff, \indexi, \Omega), \]
which we call the \emph{model-checking game} for $\calK$ and $\phi$,
is constructed in the following way, similar to the standard construction
of model-checking games for the $\mu$-calculus (c.f. \cite{FGK10}).
\end{defi}

\textbf{Positions.}
The positions of the game are pairs $(\psi, v)$,
where $\psi$ is a subformula of $\phi$, and
$v \in V$ is a location in the LHS $\calK$.
Positions $(\psi,v)$ where the top operator
of $\psi$ is $\Box, \land$, or $\nu$ belong to \pone and
all other positions belong to \pzero.
A state in the game is denoted by $s = (p, \ol{y})$, where $p=(\psi, v)$
is the position and $\ol{y}$ is the variable assignment of
the location $v$ in the underlying linear hybrid system $\calK$.

\textbf{Moves.} 
Positions of the form $(P_i, v)$ and $(y_i, v)$
are terminal positions.
From positions of the form $(\psi \land \theta, v)$,
resp. $(\psi \lor \theta, v)$, one can move to 
$(\psi, v)$ or to  $(\theta, v)$.
Positions of the form $(\Diam \psi, v)$ have
either a single successor $(-\infty)$ in case $v$ is a terminal 
location in $\calK$, or one successor $(\psi, v')$ for every
$v' \in vE$. Analogously, positions of the form $(\Box \psi, v)$
have a single successor $(\infty)$ if $vE = \emptyset$,
or one successor $(\psi, v')$ for every $v' \in vE$ otherwise.
The moves corresponding to system moves $(v,v')$ are labelled
accordingly with $\labels((v,v'))$, all other 
moves are labelled with the empty label $([0,0],(-\infty, \infty)^M, \emptyset)$
which indicates that no time passes, there are no constraints on the variables
and no variable is reset.
Fixed-point positions $(\mu X. \psi, v)$, resp. $(\nu X. \psi, v)$
have a single successor $(\psi, v)$.
Whenever one encounters a position where the fixed-point variable stands alone,
\ie $(X, v')$, the play goes back to the corresponding definition,
to $(\psi, v')$.

\textbf{Payoffs.} 
The function $\indexi$ assigns $\sem{P_i}(v)$
to all positions $(P_i, v)$, $\pm \infty$ to all
positions $(\pm \infty)$
and $y_i$ to positions $(y_i, v)$.
To discourage the players from ending the game at any other position
than a terminal one, $\indexi$ assigns all other positions
outcome $-\infty$ for \pzero's positions or $\infty$
for \pone's positions.
The payoff $\pay(\pi)$ of a play $\pi$ is calculated
using $\indexi$ and the priorities as stated before.

\textbf{Priorities.} 
The priority function $\Omega$ is defined as in the classical case
using the alternation level of the fixed-point variables,
see \eg \cite{Graedel03}.
Positions $(X, v)$ get a lower priority
than positions $(X', v')$ if $X$ has a lower alternation level than $X'$.
The priorities are then adjusted to have the right parity, such that
an even value is assigned to all positions $(X, v)$ where $X$ is a 
$\nu$-variable and an odd value to those where $X$ is a $\mu$-variable.
The maximum priority, equal to the alternation depth of the formula,
is assigned to all other positions.

\begin{exa}
We continue our example of the leaking gas burner and present in
Figure \ref{fig-mcgame-example} the model-checking
game for the previously introduced system and formula.
In this interval parity game, ellipses depict positions of \pzero and rectangles
those of \pone. In this game, all priorities are odd (and therefore
omitted), \ie infinite plays are bad for \pzero.
There is only one position with a constraints on variable $y_0$ and
in only two positions a choice about the time that passes can be made.
Both of these positions belong to \pzero in this example and are labelled
with the corresponding intervals below (and in both $y_0$ is also reset).
In terminal nodes, either the variable $y_0$ or the predicate
$P$ is evaluated for the payoff (this choice can be made by \pone in this example).
The value of the game is $1$, as is the value of the formula on the
system starting from either node, and an optimal strategy for
\pzero is picking $1$ from $[0,1]$ and then leaving the cycle
where \pone is forced to choose between the evaluation of
$y_0$ or $P$ at $v_1$. Since he is minimising, he will choose
to evaluate $y_0$.

\begin{figure}
\begin{center}
\begin{tikzpicture}[scale=0.75] \small
\node [ellipse,draw,inner sep=2pt] (a) at (0,2.5)
  {$\scriptstyle{\mu X.(\Diamond X \lor (y_0 \land P)), v_0}$};
\node [ellipse,draw] (b) at (0,1) 
  {$\scriptstyle{\Diamond X \lor ((y_0 \land P), v_0}$};
\node [draw] (e) at (-4,1) {$\scriptstyle{y_0 \land P, v_0}$};
\node [ellipse,draw] (f) at (-4, 2.5) {{$\scriptstyle{y_0}$}};
\node [ellipse,draw] (g) at (-4,-0.5) {{$\scriptstyle{-\infty}$}};
\node [ellipse, draw] (x) at (0,-0.5) {$\scriptstyle{\Diamond X, v_0}$};
\node [ellipse,draw] (d) at (0,-2) {$\scriptstyle{X, v_1}$};
\path[->] (a) edge []node {} (b);
\path[->] (b) edge []node {} (x);
\path[->] (x) edge []node[right] {{$\scriptstyle{R=\{y_0\}, [0,1]}$}} (d);
\path[->] (b) edge []node {} (e);
\path[->] (e) edge []node {} (f);
\path[->] (e) edge []node {} (g);
\node [ellipse,draw,inner sep=2pt] (a2) at (5,-2)
  {$\scriptstyle{\mu X.(\Diamond X \lor (y_0 \land P)), v_1}$};
\node [ellipse,draw] (b2) at (5,-0.5)
  {$\scriptstyle{\Diamond X \lor ((y_0 \land P), v_1}$};
\node [ellipse, draw] (x2) at (5,1) {$\scriptstyle{\Diamond X, v_1}$};
\node [ellipse,draw] (d2) at (5,2.5) {$\scriptstyle{X, v_0}$};
\node [draw] (e2) at (9,-0.5) {$\scriptstyle{y_0 \land P, v_1}$};
\node [ellipse,draw] (f2) at (9, 1) {{$\scriptstyle{y_0}$}};
\node [ellipse,draw] (g2) at (9,-2) {{$\scriptstyle{\infty}$}};
\path[->] (a2) edge []node {} (b2);
\path[->] (b2) edge []node {} (x2);
\path[->] (x2) edge []node[right] {{$\scriptstyle{R=\{y_0\}, y_0 \in [30,40],[0,\infty)}$}} (d2);
\path[->] (b2) edge []node {} (e2);
\path[->] (e2) edge []node {} (f2);
\path[->] (e2) edge []node {} (g2);
\path[->] (d) edge []node {} (a2);
\path[->] (d2) edge []node {} (a);
\end{tikzpicture}
\end{center}
\caption{Model-checking game for $\mu X. (\Diam X \lor (y_0 \land P))$ on initialised leaking gas burner.}
\label{fig-mcgame-example}
\end{figure}
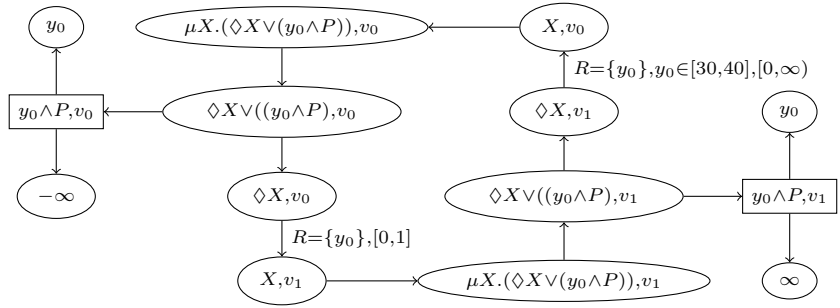

\end{exa}

It has been shown in \cite{FGK10} that quantitative parity games of any size
are determined and that they are model-checking games for \Qmu. These results
translate to interval parity games and we can conclude the following.

\begin{thm} \label{mccorrect}
Every interval parity game is determined and
for every formula $\phi$ in \Qmu, linear hybrid system $\calK$,
and a location $v$ of $\calK$, it holds that
\[ \val \mc[\calK,\phi]((\phi, v), \ol{0}) = {\sem \phi}^\calK(v, \ol{0}). \]
\end{thm}

\begin{proof}
Determinacy of an interval parity game $\calG$ follows directly from
the determinacy of the infinite QPG $\calG^*$ used to define $\calG$.

Let $\phi$ be a \Qmu-formula and $\calK$ a linear hybrid system.
Let $S(\calK) = (S, E^S)$ be the state graph of $\calK$, where $S$
is the set of all states, and $(s,s')\in E^S$ iff $s'\in \successor(s)$
in $\calK$. Let $\calK^* = (S, E^S, P_{y_0} \ldots P_{y_M})$ be the quantitative
transition system with predicates $P_{y_i}$ where $P_{y_i}(v, \ol{a}) = a_i$.
Let us also rewrite the formula $\phi$ into a formula without 
variables, $\phi^*$, by replacing each occurrence of $y_i$ by the 
corresponding $P_{y_i}$.

Applying the model-checking Theorem 12 from \cite{FGK10} we conclude
that for all $v \in \calK^*$ it holds
$\val \mc[\calK,\phi]^*(\phi, v) = {\sem \phi^*}^{\calK^*}(v)$,
\ie that $\mc[\calK, \phi]^*$ is the model-checking game for
$\calK^*$ and $\phi^*$. Finally, by definition of IPGs on the one hand and
the semantics of $\Qmu$ on the other, it follows that for all $\ol{x}$
\[ 
    \val \mc[\calK,\phi]((\phi, v), \ol{x}) = {\sem \phi}^\calK(v, \ol{x}). \]
~\\ \vskip -4.2em
\end{proof}

\section{Basic Properties of Interval Games}\label{sec_basic}
In this section, we first give a brief example that illustrates the difference between
interval games and timed games. Then, we show how to transform an initialised interval game over 
$\Qinf$ into an easier game over $\Zinf$ in which the all evolution rates are one.

At first sight, interval games seem to be very similar to
timed games. Simple timed games are solved by playing
on the region graph and can thus be discretised.
To stress that quantitative payoffs indeed make a difference,
we present in Figure \ref{fig-game-half-example} an initialised interval parity 
game with the interesting property that it is not optimal to play integer
values, even though the underlying system is over $\Zinf$. This simple game
contains only one variable (a clock) and has no constraints on this variable
in any of the transitions, so only the time intervals are shown.
Also, as infinite plays are not possible, the priorities are omitted,
as well as the indices of non-terminal positions (they are chosen to be
unfavourable for the current player such that she has to continue playing). 
The payoff rule specifies the outcome of a play $\pi$
ending in $v_2$ as $p(\pi)=y_0 -1$ and in $v_3$ as $p(\pi)=-y_0$.
This game illustrates that it may not be optimal to 
play integer values since choosing time $\frac{1}{2}$ in the first move 
is optimal for \pzero. This move guarantees an outcome of $-\frac{1}{2}$
which is equal to the value of the game.

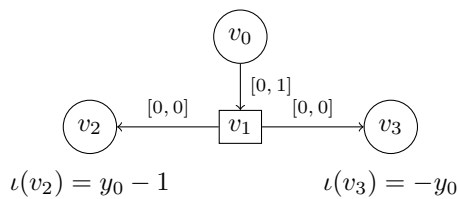
\begin{figure}[h]
\begin{center}
\begin{tikzpicture}[scale=1]
\small
\node[draw,circle] (v0) at (0, 2.2) {$v_0$};
\node[draw] (v1) at (0, 1) {$v_1$};
\node[draw,circle](e1) at (-2, 1) {$v_2$};
\node[]() at (-2, 0.25) {\small{$\iota(v_2) = y_0 -1$}};
\node[draw,circle](e2) at (2, 1) {$v_3$};
\node[]() at (2, 0.25) {\small{$\iota(v_3) = -y_0$}};
\path[->] (v0)  edge node[right] {\tiny{$[0,1]$}} (v1);
\path[->] (v1)  edge node[above] {\tiny{$[0,0]$}} (e1);
\path[->] (v1)  edge node[above] {\tiny{$[0,0]$}} (e2);
\end{tikzpicture}
\end{center}
\caption{Game with integer coefficients and non-integer value.}
\label{fig-game-half-example}
\end{figure}

\subsection{Flattening Initialised Interval Games} \label{subsec_flat}
So far, we have considered games where the values of variables can change at
different rates during the time spent in locations. In this section,
we show that for initialised games it is sufficient to look at easier
games where all rates are one, similar to timed games but with more
complex payoff rules. We call these games flat and show that for every 
initialised IPG we can construct a flat IPG with the same value.
To do so, we have to consider the regions where the coefficients do not
change and rescale the constraints and payoffs accordingly.

For an interval $I = [i_1, i_2]$, we denote by $q \cdot I$
and $q + I$ the intervals $[q \cdot i_1, q \cdot i_2]$ and $[q+i_1, q+i_2]$
respectively, and do analogously for open intervals. 

\begin{defi}
An interval parity game 
$\calG = (V_0, V_1, E, \labels, \coeff,\indexi, \Omega)$
is \emph{flat} if and only if $\coeff_i(v)=1$
for all $v \in V$ and $i=1 \ldots M$.
\end{defi}

\begin{lem} \label{flat}
For each initialised interval parity game $\calG$ there exists
a flat game $\calG'$ with the same value.
\end{lem}

\begin{proof}
Let $\calG = (V_0, V_1, E, \labels, \coeff,\indexi, \Omega)$
be an initialised interval parity game.
We construct a corresponding flat game 
$\calG'=(V_0, V_1, E, \labels', \coeff',\indexi', \Omega)$ in the following way:
For a position $v \in V = V_0 \cupdot V_1$ and each variable $y_i$, such that 
$\coeff_i(v)=a_i$, $\indexi(v)=a \cdot y_i + b$ and an outgoing edge $(v,w)$ with $C_i=[c_0,c_1]$
we have in the corresponding flat game:
\begin{iteMize}{$\bullet$}
\item $\coeff_i'(v)=1$
\item $C_i' \in \labels'(v,w) = [ \frac {c_0} {a_i} , \frac {c_1} {a_i}] = \frac 1 {a_i} C_i$
\item $\indexi'(v)=a_i \cdot a \cdot y_i + b$
\end{iteMize}
Note that we only change the functions $\coeff, \labels$ and $\indexi$.
We will show that for every play $\pi$ from a starting state $s$ consistent
with $\sigma$ and $\rho$, we can construct strategies $\sigma'$, $\rho'$,
such that $\pi'(\sigma', \rho',s')$ visits the same locations as $\pi$ and
$\pay(\pi)=\pay(\pi')$. Before we proceed with the proof, notice that it is
essential that $\calG$ is an initialised game. Intuitively, the value of $y_i$
in $\calG'$ is the value of $y_i$ in $\calG$ divided by the coefficient $a_i$
of the current position. When the position changes, it is thus crucial that
$a_i$ does \emph{not} change, except if $y_i$ is reset -- exactly what is
required from an initialised game.

The proof proceeds by induction on the length of the plays.
First, if $s_0=(v_0, \ol 0)$ is a state belonging to \pzero and
$\sigma(s_0)=s_1=(v_1, \ol x)$ and $s_0'=(v_0, \ol 0)$, then in $\calG'$
we define $\sigma'(s_0')=s_1'$, where $s_1'=(v_1,\ol y')$, such that
$y_i'=\frac {y_i} {a_i}$ for any $y_i \nin R \in \labels(v_0,v_1)$.
Since $(s_0,s_1)$ is allowed in $\calG$, this means that for all
$y_i \nin R \in \labels(v_0,v_1)$, we have
$y_i \in C_i=[c_0,c_1] \in \labels(v_0,v_1)$. It follows that
$\frac {c_0} {a_i} \leq y_i' = \frac {y_i} {a_i} \leq \frac {c_1} {a_i}$
for all $y_i \nin R$ and therefore $(s_0',s_1')$ is allowed in $\calG'$.
Also $\pay(s_1)= \indexi (v_1) = a \cdot y_i + b$ and therefore the payoff is
equal to $\pay(s_1')= \indexi'(v_1') = a_i \cdot a \cdot \frac {y_i} {a_i} + b$.

Let $s_0 \ldots s_k$ and $s_0' \ldots s_k'$ be finite histories in $\calG$ and
$\calG'$, such that they visit the same locations and $\pay(\pi)=\pay(\pi')$.
Then, if $s_k=(v_k, \ol y)$ is a state belonging to \pzero and
$\sigma(s_k)=s_{k+1}=(v_{k+1}, \ol y)$ and $s_k'=(v_k, \ol z)$,
then in $\calG'$ we define $\sigma'(s_k')=s_{k+1}'$, where
$s_{k+1}'=(v_k,\ol w)$, such that $w_i=t$ where $t_i= \frac {y_i} {a_i}$
for any $y_i \nin R \in \labels(v_k,v_{k+1})$.
Since $(s_k,s_{k+1})$ is allowed in $\calG$, this means that for
all $y_i \nin R$, $y_i \in C_i=[c_0,c_1] \in \labels(v_k,v_{k+1})$.
As $\frac {c_0} {a_i} \leq w_i= \frac {y_i} {a_i} \leq \frac {c_1} {a_i}$ for
all $y_i \nin R$, we get that $(s_k',s_{k+1}')$ is allowed in $\calG'$.
Also $\pay(s_k)=\indexi(v_k)=a \cdot y_i + b$ and therefore the payoff is equal
to $\pay(s_{k+1}')=\indexi'(v_{k+1}')=a_i \cdot a \cdot w_i + b =
    a_i \cdot a \cdot \frac {y_i} {a_i} + b$.

The cases for \pone are analogous. Note that, for infinite plays,
we also have the same payoff, since for the payoff of infinite games only
the locations (and their priorities) matter. Since we can construct, for
each pair of strategies in $\calG$, the corresponding strategies in $\calG'$,
and those yield a play with the same payoff, the values of the two games
are equal.
\end{proof}

Consequently, from now on we only consider flat interval parity games and
therefore omit the coefficients, as they are all equal to one.

\subsection{Multiplying Interval Games}\label{subsec_mult}

\begin{defi}
For a flat IPG $\calG = (V_0, V_1, E, \labels, \indexi, \Omega)$ and
a value $q \in \setQ$, we denote by
$q \cdot \calG = (V, E, \labels', \indexi', \Omega)$ the IPG where 
$\indexi'(v) = a \cdot y_i + q \cdot b$ iff $\indexi(v) = a \cdot y_i + b$ 
for all $v \in V$, and  $(I', \ol{C'}, R) \in \lambda'((v,w))$ iff
$(I, \ol{C}, R) \in \lambda((v,w))$ with $I'= q \cdot I$ and
${C'}_i= q \cdot C_i$ for all $(v,w) \in E$.
\end{defi}

Intuitively, this means that all endpoints in the time intervals (open and closed),
and the constraints, and all additive values in the payoff function $\indexi$
are multiplied by $q$. The values of $q \cdot \calG$ are also equal to
the values of $\calG$ multiplied by $q$.

\begin{lem} \label{mult_games}
For every IPG $\calG$ over $\Qinf$ and $q \in \setQ, q \neq 0$ it holds in all states
$s$ that $q \cdot \val \calG (s) = \val \  q \cdot \calG(q \cdot s)$.
\end{lem}

\begin{proof}
We denote by $q \cdot \sigma$ the strategy with 
$q \cdot \sigma(q \cdot h) = q \cdot s'$ iff $\sigma(h)= s'.$
The mapping of $\calG$ with strategies for both players $\sigma$ and $\rho$
to $q \cdot \calG$ with $q \cdot \sigma$ and $q \cdot \rho$ is a bijection 
(in the reverse direction take $\frac 1 q$). We also have 
$q \cdot \pay_{\calG} (\pi(\sigma, \rho, s)=s_0s_1 \ldots s_k) = q \cdot (a \cdot y_i + b)$
where $\indexi(loc(s_k))=(a,i,b)$ which is equal to 
$\pay_{q \cdot \calG} (\pi(q \cdot \sigma, q \cdot \rho, q \cdot s)= q \cdot s_0 \ldots q \cdot s_k)= a \cdot (q \cdot y_i) + q \cdot b$
for all finite plays $\pi$. Therefore, we know that 
$\inf_{\rho} q \cdot \pay (\pi(\sigma, \rho, s) = 
 \inf_{q \cdot \rho}\pay (\pi(q \cdot \sigma, q \cdot \rho, q \cdot s)$ 
and the same holds for the supremum and thus we get the desired result.
\end{proof}

Note that all multiplicative factors in $\iota$ are the same in $\calG$ and
in $q \cdot \calG$. Moreover, if we multiply all constants in $\iota$ in a game
$\calG$ (both the multiplicative and the additive ones) by a positive value $r$,
then the value of $\calG$ will be multiplied by $r$, by an analogous argument as
above. Thus, if we first take $r$ as the least common multiple of all denominators
of multiplicative factors in $\iota$ and multiply all $\iota$ constants as above,
and then take $q$ as the least common multiple of all denominators of endpoints
in  the intervals and additive factors in the resulting game $\calG$ and build
$q \cdot \calG$, we can conclude the following.

\begin{cor} \label{mult_games_col}
For every finite IPG $\calG$ over $\Qinf$, there exists an IPG $\calG'$ over $\Zinf$ and
$q, r \in \bbZ$ such that $\val\calG(s) = \frac{\val\calG'(q \cdot s)}{q \cdot r}$.
\end{cor}

From now on we assume that every IPG we investigate is a flat game over $\Zinf$
when not explicitly stated otherwise.

\section{Discrete Strategies} \label{sec_ds}

Our goal in this section is to show that it suffices to use a simple kind
of (almost) discrete strategies to approximate the value of flat 
interval parity games over $\Zinf$.
To this end, we define an equivalence relation between states
whose variables belong to the same $\bbZ$ intervals. This equivalence,
resembling the standard methods used to build the region graph from
timed automata, is a technical tool needed to compare the values of
the game in similar states. 

We use the standard meaning
of $\floor r$ and $\ceil r$, and denote by $\{r\}$ the number $r - \floor{r}$
and by $[r]$ the pair $(\floor r , \ceil r)$. Hence, when writing
$[r]=[s]$, we mean that $r$ and $s$ lie in between the same integers.
Note that if $r \in \bbZ$ then $[r] = [s]$ implies that $r = s$.

\begin{defi}
We say that two states $s$ and $t$ in an IPG are equivalent, $s \sim t$,
if they are in the same location, $\loc(s) = \loc(t)$, and for all
$i, j \in \{1, \ldots, K\}$:
\begin{iteMize}{$\bullet$}
\item $[\var_i(s)] = [\var_i(t)]$, and
\item if $\{ \var_i(s) \} \leq \{ \var_j(s) \}$ then $\{ \var_i(t) \} \leq \{ \var_j(t) \}$.
\end{iteMize}
\end{defi}

Intuitively, all variables lie in the same integer intervals and the order
of fractional parts is preserved. In particular, it follows that all
integer variables are equal.
The following technical lemma allows for the shifting of moves between $\sim$-states.

\begin{lem} \label{shift-move}
Let $s$ and $s'$ be two states in a flat IPG over $\bbZ$ such that $s \sim s'$.
If a move from $s$ to $t$ is allowed by a label $l = (I, \ol{C}, R)$,
then there exists a state $t'$, the move to which from $s'$ is allowed by
the same label $l$ and $t' \sim t$.
\end{lem}

\begin{proof}
If $R = \{1, \ldots, K\}$ then let $t' = t$.
As $s \sim s'$, the same constraints are satisfied by $s$ and $s'$ and thus
the move from $s'$ to $t' = t$ is allowed by the same label.

If $R \neq \{1, \ldots, K\}$ then let $w = t -_R s \in I$ be the increment
chosen during the move. If $w \in \bbZ$ we let $t' = s'+w$, the conditions
follow from the assumption that $s \sim s'$ again.

If $w \not\in \bbZ$, let $i$ be the index of a non-reset variable
with the smallest fractional part in $t$, i.e.
$\{ \var_i(t) \} \leq \{ \var_j(t) \}$ for all $j \not\in R$.
To construct $t'$, we must choose $w'$ with $[w'] = [w]$ which makes
$\var_i(s'+w')$ the one with smallest fractional part. 

\emph{Case 1}: $\{ \var_i(t) \} \geq \{w\}$.\\
  In this case, for all non-reset variables $j$, holds
  $\{ \var_j(t) \} \geq \{w\}$, intuitively meaning
  that no variable ``jumped'' above an integer due to $\{w\}$.
  Let $l$ be the variable with maximum fractional part in $s'$
  (and thus, by definition of $\sim$, also in $s$ and in this case in $t$). Set 
  \[ w' = \floor{w} + 0.9 \cdot \left(\ceil{\var_l(s')} - \var_l(s')\right) .\]
  Clearly $[w'] = [w]$ and indeed, we preserved the order of
  fractional parts and integer intervals, thus $\sim$ is preserved.

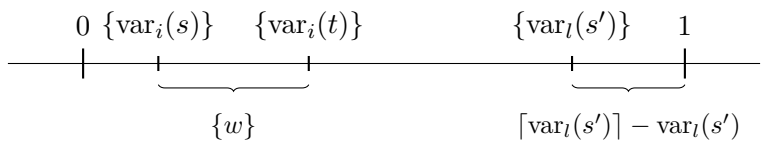
\begin{figure}[h]
\begin{center}
\begin{tikzpicture}
\draw (-5, 0) -- (5, 0);
\draw[thick] (-4, 0.2) -- (-4, -0.2);
\node (a) at (-4, 0.5) {$0$};

\draw[thick] (-3, 0.1) -- (-3, -0.1);
\node (b0) at (-3, 0.5) {$\{\var_i(s)\}$};

\draw[thick] (-1, 0.1) -- (-1, -0.1);
\node (b1) at (-1, 0.5) {$\{\var_i(t)\}$};

\draw [snake=brace, mirror snake, raise snake=8pt]
   (-3,0) to node [below=15pt] {\small{$\{w\}$}} (-1,0);

\draw[thick] (2.5, 0.1) -- (2.5, -0.1);
\node (b3) at (2.5, 0.5) {$\{\var_l(s')\}$};

\draw [snake=brace, mirror snake, raise snake=8pt] (2.5,0) to
  node [below=15pt] {\small{$\ceil{\var_l(s')} - \var_l(s')$}} (4,0);

\draw [thick] (4, 0.2) -- (4, -0.2);
\node (d) at (4, 0.5) {$1$};
\end{tikzpicture}
\end{center}
\caption{Lemma~\ref{shift-move} Case 1}
\label{fig-case-1}
\end{figure}

\emph{Case 2}: $\{ \var_i(t) \} < \{w\}$ and for all $j \not\in R$
$\{ \var_j(s') \} \geq \{ \var_i(s') \}$.\\
  In this case, for all non-reset variables $j$, holds
  $\{ \var_j(t) \} \leq \{w\}$, intuitively meaning
  that all variables ``jumped'' above an integer due to $\{w\}$.
  Let $l$ be the variable with maximum fractional part in $s'$
  (and thus also in $s$). Let
  \[ \delta = 0.9 \cdot \min\left(
        \{ \var_i(s') \}, \left( \ceil{\var_l(s')} - \var_l(s') \right) 
      \right)\]
  be a number smaller than both $\{ \var_i(s') \}$ and 
  $\ceil{\var_l(s')} - \var_l(s')$. We set
  \[ w' = \floor{w} + \ceil{\var_i(s')} - \var_i(s') + \delta.\]
  By the first assumption on $\delta$ we have $[w'] = [w]$ and both the order
  of fractional parts and integer bounds in $t'$ are the same as in $t$, since
  \[ \ceil{\var_l(t')} = \ceil{\var_l(s'+w')} \leq
     \ceil{\var_l(s') + \floor{w} + 1 + \delta} = \ceil{\var_l(t)} \]
  by the second assumption on $\delta$. The inequality in the other
  direction holds as well, and we get that $t' \sim t$ as required.

\begin{figure}[h]
\begin{center}
\begin{tikzpicture}
\draw (-5, 0) -- (5, 0);
\draw[thick] (2, 0.2) -- (2, -0.2);
\node (a) at (2, 0.5) {$1$};

\draw[thick] (-4, 0.1) -- (-4, -0.1);
\node (b0) at (-4, 0.5) {$\{\var_i(s')\}$};

\draw[thick] (-2, 0.1) -- (-2, -0.1);
\node (b0) at (-2, 0.5) {$\{\var_i(s)\}$};

\draw[thick] (0.5, 0.1) -- (0.5, -0.1);
\node (b0) at (0.5, 0.5) {$\{\var_l(s')\}$};

\draw [snake=brace, mirror snake, raise snake=2pt]
   (0.5,0) to node [below=2pt] {\small{$\delta$}} (1,0);

\draw [snake=brace, mirror snake, raise snake=12pt] (-4,0) to
  node [below=15pt] {\small{$\ceil{\var_i(s')} - \var_i(s')$}} (2,0);

\draw[thick] (4, 0.1) -- (4, -0.1);
\node (b3) at (4, 0.5) {$\{\var_i(t)\}+1$};
\end{tikzpicture}
\end{center}
\caption{Lemma~\ref{shift-move} Case 2}
\label{fig-case-2}
\end{figure}
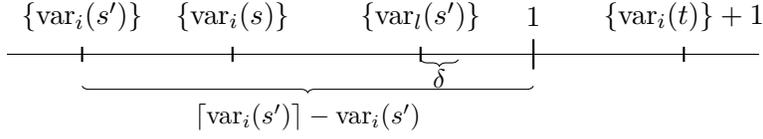

\emph{Case 3}: $\{ \var_i(t) \} < \{w\}$ and there exists $j \not\in R$
  with $\{ \var_j(s') \} < \{ \var_i(s') \}$.\\
  In this case let $l$ be the variable with maximum fractional part in $t$,
  \ie the last one which did not ``jump'' above an integer due to $\{w\}$.
  The variable with next bigger fractional part in $s$ (and by $\sim$ also
  in $s'$) is $\var_i(s)$, as depicted in Figure~\ref{fig-case-3-s}.

\begin{figure}[h]
\begin{center}
\begin{tikzpicture}
\draw (-5, 0) -- (5, 0);
\draw[thick] (2, 0.2) -- (2, -0.2);
\node (a) at (2, 0.5) {$1$};

\draw[thick] (-4, 0.1) -- (-4, -0.1);
\node (b0) at (-4, 0.5) {$\{\var_l(s)\}$};

\draw[thick] (-1, 0.1) -- (-1, -0.1);
\node (b0) at (-1, 0.5) {$\{\var_i(s)\}$};

\draw[thick] (0.5, 0.1) -- (0.5, -0.1);
\node (b0) at (0.5, 0.5) {$\{\var_l(t)\}$};

\draw [snake=brace, mirror snake, raise snake=6pt]
   (-4,0) to node [below=12pt] {\small{$\{ w \}$}} (0.5,0);

\draw[thick] (3.5, 0.1) -- (3.5, -0.1);
\node (b3) at (3.5, 0.5) {$\{\var_i(t)\}+1$};
\end{tikzpicture}
\end{center}
\caption{Lemma~\ref{shift-move} Case 3 for $s$}
\label{fig-case-3-s}
\end{figure}

  To transfer the move to $s'$, consider these two variables in $s'$ as
  depicted in Figure~\ref{fig-case-3-sprime} and let
  $\delta = \{\var_i(s')\} - \{\var_l(s')\}$.

\begin{figure}[h]
\begin{center}
\begin{tikzpicture}
\draw (-5, 0) -- (5, 0);
\draw[thick] (2, 0.2) -- (2, -0.2);
\node (a) at (2, 0.5) {$1$};

\draw[thick] (-4, 0.1) -- (-4, -0.1);
\node (b0) at (-4, 0.5) {$\{\var_l(s')\}$};

\draw[thick] (-1, 0.1) -- (-1, -0.1);
\node (b0) at (-1, 0.5) {$\{\var_i(s')\}$};

\draw [snake=brace, mirror snake, raise snake=6pt]
   (-4,0) to node [below=12pt] {\small{$\delta$}} (-1,0);

\draw [snake=brace, mirror snake, raise snake=6pt]
   (-1,0) to node [below=12pt] {\small{$\ceil{\var_i(s')} - \var_i(s')$}} (2,0);
\end{tikzpicture}
\end{center}
\caption{Lemma~\ref{shift-move} Case 3 for $s'$}
\label{fig-case-3-sprime}
\end{figure}

  We set $w' = \floor{w} + \ceil{\var_i(s')} - \var_i(s') + 0.9 \cdot \delta.$
  Again $[w'] = [w]$ and clearly $i$ is the variable with smallest fractional
  part in $t'$ by construction. As $s \sim s'$, the order of fractional
  parts in $t$ and in $t'$ is the same, and the integer bounds as well,
  thus $t \sim t'$.
\end{proof}

\subsection{Choosing Discrete Moves}\label{subsec_choose}

Knowing that we can shift a single move and preserve $\sim$-equivalence,
we proceed to show that for IPGs over $\Zinf$, fully general strategies are
not necessary. In fact, we can restrict ourselves to discrete strategies
and, using this, reduce the games to discrete systems.
Intuitively, a discrete strategy keeps the maximal distance of 
all variable valuations to the closest integer small.

However, for the purposes of constructing an inductive proof 
of existence of a good discrete strategy,
it is not convenient to work, for a state $s$, simply with the maximal distance
\[ \max_i \{ \min \{ \var_i(s) - \floor{\var_i(s)},
                   \ceil{\var_i(s)} - \var_i(s) \} \}. \]
The reason is that for some moves it is impossible to keep this distance
small for each variable and to go to an equivalent state as illustrated
in Figure \ref{fig-prob}. In the depicted situation, if we move $y_1$ within
$\epsilon$-neighbourhood of $\bbZ$ (below $z$ and $z-1$ depict integers),
then $y_0$ leaves it.

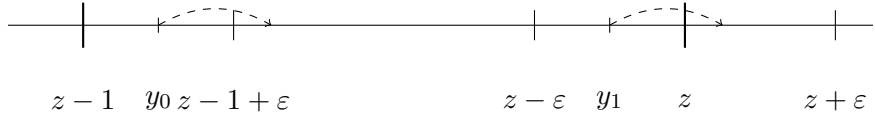
\begin{figure}[h]
\begin{center}
\begin{tikzpicture}
\draw (-5, 0) -- (6.5, 0);
\draw[thick] (-4, 0.3) -- (-4, -0.3);
\node (z0) at (-4, -1) {$z-1$};
\draw (-3, 0.1) -- (-3, -0.1);
\node (z3) at (-3, -1) {$y_0$};
\draw (-2, 0.2) -- (-2, -0.2);
\node (z2) at (-2, -1) {$z-1+\epsilon$};
\draw [thick] (4, 0.3) -- (4, -0.3);
\node (z1) at (4, -1) {$z$};
\draw (2, 0.2) -- (2, -0.2);
\node (e) at (2, -1) {$z-\epsilon$};
\draw (3, 0.1) -- (3, -0.1);
\node (y) at (3, -1) {$y_1$};
\draw (6, 0.2) -- (6, -0.2);
\node (z4) at (6, -1) {$z+\epsilon$};
\path[->,dashed] (3,0) edge [bend left] (4.5,0);
\path[->,dashed] (-3,0) edge [bend left] (-1.5,0);
\end{tikzpicture}
\end{center}
\caption{Move where standard distance is necessarily increased.}
\label{fig-prob}
\end{figure}

\noindent To give a more suitable notion of distance for a state,
let us, for $r \in \setR$, define
\begin{displaymath}
\di(r) = \left\{
\begin{array}{ll} 
r-\ceil{r} & 
    \text{ if } |r-\ceil{r}| \leq |r-\floor{r}|; \\
r-\floor{r} & 
    \text{ otherwise.}
\end{array} \right.
\end{displaymath}
This function gives the distance to the closest integer, except
that it is negative if the closest integer is greater than $r$,
\ie if the fractional part of $r$ is $> \frac{1}{2}$.
as depicted in Figure \ref{fig-dist}.

\begin{figure}[h]
\begin{center}
\begin{tikzpicture}
\draw (-5, 0) -- (5, 0);
\draw[thick] (-4, 0.3) -- (-4, -0.3);
\node (a) at (-4, 0.5) {$\floor{r}$};
\draw (-3, 0.1) -- (-3, -0.1);
\node (b) at (-3, 0.5) {$r$};
\draw [thick] (0, 0.3) -- (0, -0.3);
\node (c) at (0, 0.5) {$\ceil{r}= \floor{s}$};
\draw [thick] (4, 0.3) -- (4, -0.3);
\node (d) at (4, 0.5) {$\ceil{s}$};
\draw (3, 0.1) -- (3, -0.1);
\node (e) at (3, 0.5) {$s$};
\draw [snake=brace, mirror snake, raise snake=10pt](3,0) to node [below=15pt]
   {\small{$\di(s)<0$}}(4,0);
\draw [snake=brace, mirror snake, raise snake=10pt](-4,0) to node [below=15pt]
   {\small{$\di(r)>0$}}(-3,0);
\end{tikzpicture}
\end{center}
\caption{Notation for distances between real numbers and integers.}
\label{fig-dist}
\end{figure}
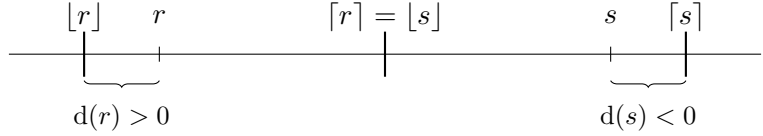

Please observe that for two real numbers $a,b \in \setR_+$,
it follows that
\[ |\di(a + b)| \leq |\di(a)| + |\di(b)|.\]
Also, we observe that 
\begin{iteMize}{$\bullet$}
\item if $|\di(a) + \di(b)|< \frac 1 2$, then $\di(a + b) = \di(a) + \di(b)$;
\item otherwise, if  $\di(a),\di(b) = \frac 1 2$ or $\di(a),\di(b)=0$, then $\di(a + b)=0$;
\item otherwise, if $\di(a),\di(b)>0$, then $\di(a + b) = \di(a) + \di(b) - 1 < 0$;
\item if $\di(a),\di(b)<0$, then $\di(a + b) = \di(a) + \di(b) + 1 > 0$.
\end{iteMize}

For a state $s$, we use the abbreviation $\di_i(s) = \di(\var_i(s))$.
We denote by $\dleft(s)= \min_{i=1 \ldots k}\{\di_i(s)\}$ and
$\dright(s)= \max_{i=1 \ldots k}\{\di_i(s)\}$ the smallest and biggest
of all values $\di_i(s)$, and additionally we define the total distance
as follows
\begin{displaymath} \label{max-total-dist}
\dih(s) = \left\{
\begin{array}{ll} 
|\dleft(s)| & 
    \text{ if } \di_i (s)\leq 0 \text{ for all } i \in \{1, \ldots, k\},\\
\dright(s) & 
    \text{ if } \di_i(s) \geq 0 \text{ for all } i \in \{1, \ldots, k\},\\
|\dleft(s)|+ \dright(s) & 
    \text{ otherwise.}
\end{array} \right.
\end{displaymath}

This is illustrated in Figure~\ref{fig-total-dists}, where $k$ stands for 
an integer and $y_0$ to $y_2$ stand for the fractional parts of the values 
of the respective variables.
In this example, $y_0$ has the smallest fractional part, \ie the biggest one bigger
than $\frac 1 2$ and $y_2$ has the biggest fractional part (less than $\frac 1 2$).

\begin{figure}[b]
\begin{center}
\begin{tikzpicture}
\draw (-5, 0) -- (5, 0);
\draw[thick] (-4, 0.2) -- (-4, -0.2);
\node (a) at (-4, -0.5) {$k - \frac{1}{2}$};
\draw[thick] (-3, 0.1) -- (-3, -0.1);
\node (b0) at (-3, 0.5) {$y_0$};
\draw[thick] (1.5, 0.1) -- (1.5, -0.1);
\node (b1) at (1.5, 0.5) {$y_1$};
\draw[thick] (2.5, 0.1) -- (2.5, -0.1);
\node (b3) at (0.4, 0.5) {$y_3$};
\draw[thick] (0.4, 0.1) -- (0.4, -0.1);
\node (b2) at (2.5, 0.5) {$y_2$};
\draw [thick] (0, 0.2) -- (0, -0.2);
\node (c) at (0, -0.5) {$k$};
\draw [thick] (4, 0.2) -- (4, -0.2);
\node (d) at (4, -0.5) {$k+\frac{1}{2}$};
\draw [snake=brace, mirror snake, raise snake=8pt]
   (0,0) to node [below=15pt] {\small{$\dright$}} (2.5,0);
\draw [snake=brace, mirror snake, raise snake=8pt]
   (-3,0) to node [below=15pt] {\small{$- \dleft$}} (0,0);
\draw [snake=brace, raise snake=8pt]
   (-3,0) to node [above=15pt] {\small{$\dih$}} (2.5,0);
\end{tikzpicture}
\end{center}
\caption{Maximal, minimal and total distances for a state.}
\label{fig-total-dists}
\end{figure}
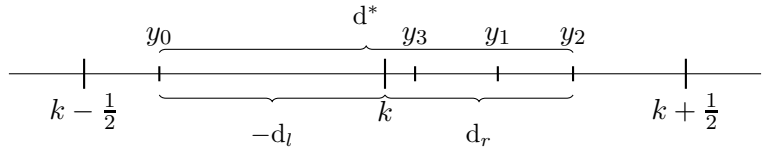

First, we will prove that we can always correct a strategy that makes
one step which is not $\epsilon$-discrete. By doing so, we will
guarantee that we reach a state with the same location that is allowed
by the labelling and that the values of the variables only change
within the same intervals.

\begin{lem}\label{discrete_succ}
Let $s$ be a state with $\dih(s) \leq \frac 1 4$ and $t$ be a successor
of $s$, where $(s,t)$ is allowed by $l=(I, \ol{C}, R)$.
Then, for every $0 \leq \epsilon < \dih(s)$, there exists
a successor $t'_+$ of $s$ such that 
\begin{iteMize}{$\bullet$}
\item
  $t \sim t'_+$,
\item 
  $(s, t'_+)$ is allowed by $l$, and
\item
  $\dih(t'_+) \leq \dih(s) + \epsilon$.
\end{iteMize}
\end{lem}

\begin{proof}
We assume that $\dih(t) > \dih(s) + \epsilon$,
otherwise we can take $t'_+ = t$.
Let $w \in I$ be the increase in the (non-reset) values from $s$ to $t$, \ie $w = t -_R s$.
We make a case distinction regarding the computation of $\dih(t)$.

\textit{Case 1:} $\dih(t)=|\dleft(t)|$.\\
We correct $w$ in the following way: $w' = w + \corr -\epsilon$, where
$\corr = \min\{|\dright(t)|, |d(w)|\}$ if $\di(w)<0$ and $\corr = |\dright(t)|$
otherwise.

First, we have to show that $[w'] \in [w]$ and therefore $w' \in I$.
Since $\dleft(t) = \di_i(t)= \var_i(t)$ for one $i$, 
we can conclude from $|\di(\var_i(s)+w)| \leq |\di(\var_i(s))| + |\di(w)|$
that $|\di(w)| > \epsilon$ and therefore $w' \geq w$,
hence $w' \geq \floor{w}$.
Furthermore, $w'\leq \ceil{w}$. Otherwise, if $\di(w) < 0$ then
$w' =  w + \corr - \epsilon > \ceil{w} = w + |\di(w)|$. This is a
contradiction, since by definition $\corr \leq |\di(w)|$.

If $\di(w) > 0$, we also conclude $w'\leq \ceil{w}$, since 
$\corr - \epsilon < \frac 1 2$.

Next, we have to show, that all variables that are not reset stay in the 
same interval. We consider the case, where all values of the variables are
increased, therefore we know that $\var_i(t'_+) \geq \floor{\var_i(t)}$
for all $i \not \in R$.
We now have to show that also $\var_i(t'_+) \leq \ceil{\var_i(t)}$.
Let $j$ be the index of the variable which is the closest to the integers
(in this case), \ie $j$, such that $\di(\var_j(t))=\dright(t)$.
\begin{eqnarray*}
\var_j(t'_+) & = &  \var_j(s) + w' \\ 
   & = & \var_j(s)+ w + \corr - \epsilon \\
   & = & \var_j(t) + \corr - \epsilon \\
   & < & \ceil{\var_i(t)} = \var_j(t) + |\dright(t)| \\
\end{eqnarray*}
Also, we have to show: $\dih(t'_+) \leq \dih(s) + \epsilon$.
We know that $|\dleft(t)|-|\dright(t)| \leq \dih(s)$ and
$\dih(t'_+) = |\dleft(t'_+)|=|\di(\var_j(t'_+))|$ for one $j$
and
$\var_j(t'_+) = \var_j(s) + w + \corr - \epsilon$.
Hence, $\di(\var_j(t'_+)) = \dleft(t) + \corr - \epsilon$, since
$|\dleft(t) + \corr - \epsilon| \leq \frac 1 2$.
We can conclude that $\dleft(t'_+)=\di(\var_j(t'_+)) \leq \dih(s) + \epsilon$.
\begin{figure}[h]
\begin{center}
\begin{tikzpicture}
\small
\tikzstyle{var_style}=[draw,fill=black, circle, minimum size=1mm, inner sep=0pt];
\node[] (s) at (-2,0) {$t =$};
\draw (-1, 0) -- (5, 0);
\draw[thick] (0, 0.2) -- (0, -0.2);
\draw[] (2, 0.1) -- (2, -0.1);
\node (x1c) at (3.5,0.2){\tiny{$y_1$}};
\node[var_style] (x1) at (3.5,0){};
\node (x2c) at (2.8,0.2){\tiny{$y_2$}};
\node[var_style] (x2) at (2.8,0) {};
\draw[thick] (4, 0.2) -- (4, -0.2);
\draw [snake=brace, raise snake=8pt]
    (x1) to node [above=8pt] {$|\dright|$} (4, 0);
\draw [snake=brace, raise snake=24pt]
    (x2) to node [above=24pt] {$|\dleft| > \dih(s) + \epsilon$} (4, 0);
\draw [snake=brace, raise snake=8pt]
    (0, 0) to node [above=8pt] {$1 - |\dleft|$} (x2);
\draw [snake=brace, mirror snake, raise snake=6pt]
    (x2) to node [below=8pt] {$\leq \dih(s)$} (x1);

\node[] (ss) at (-2,-3) {$t'_+ =$};
\draw (-1, -3) -- (5, -3);
\draw[thick] (0, -2.8) -- (0, -3.2);
\node (x1cp) at (3.8,-3.2){\tiny{$y_1$}};
\node[var_style] (x1p) at (3.8,-3){};
\node (x2cp) at (3.1,-3.2){\tiny{$y_2$}};
\node[var_style] (x2p) at (3.1,-3){};
\draw[] (2, -2.9) -- (2, -3.1);
\draw[thick] (4, -2.8) -- (4, -3.2);
\draw [snake=brace, raise snake=8pt]
    (x2p) to node [above=8pt] 
      {$|\dleft| \leq \dih(s) + \epsilon$} (4, -3);
\end{tikzpicture}
\end{center}
\caption{Case 1 illustration}
\label{correct1}
\end{figure}
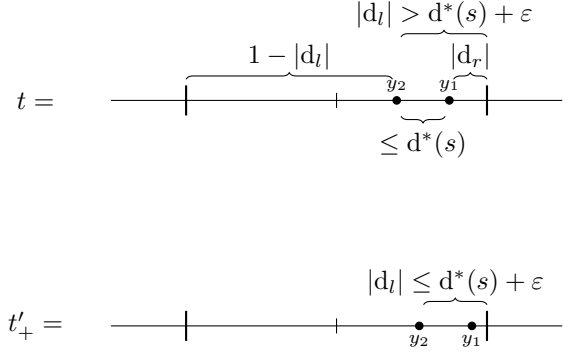

\textit{Case 2:} $\dih(t)=|\dright(t)|$.

\textit{Subcase 1: $\di(w)>0$:}\\
We correct $w$ in the following way: $w' = w + (1 - \corr) - \epsilon$, where
$\corr = \max\{|\dleft(t)|, |d(w)|\}$.

First, we have to show that $[w'] \in [w]$ and therefore $w' \in I$.
Since $\dright(t) = \di_i(t)= \var_i(t)$ for one $i$, 
we can conclude from $|\di(\var_i(s)+w)| \leq |\di(\var_i(s))| + |\di(w)|$
that $|\di(w)| > \epsilon$ and therefore $w' \geq w$,
hence $w' \geq \floor{w}$.
Furthermore, $w' \leq \ceil{w}$.
Otherwise, since $\di(w) > 0$ and we assume that
$w' =  w + (1 -\corr) - \epsilon > \ceil{w} = w + (1 - |\di(w)|)$. This is a
contradiction, since by definition $\corr \geq |\di(w)|$.

Next, we have to show, that all variables that are not reset stay in the 
same interval. We consider the case, where all values of the variables are
increased, therefore we know that $\var_i(t'_+) \geq \floor{\var_i(t)}$
for all $i \not \in R$.
We now have to show that also $\var_i(t'_+) \leq \ceil{\var_i(t)}$.
Let $j$ be the index of the variable which is the closest to the integers
(in this case), \ie $j$, such that $\di(\var_j(t))=\dleft(t)$.
\begin{eqnarray*}
\var_j(t'_+) & = &  \var_j(s) + w' \\ 
   & = & \var_j(s)+ w + (1 - \corr) - \epsilon \\
   & = & \var_j(t) + (1- \corr) - \epsilon \\
   & < & \ceil{\var_i(t)} = \var_j(t) + (1 - |\dleft(t)| \\
\end{eqnarray*}

Also, we have to show: $\dih(t'_+) \leq \dih(s) + \epsilon$.
We know that $\dright(t)-\dleft(t) \leq \dih(s)$ and
$\dih(t'_+) = |\dright(t'_+)|=|\di(\var_j(t'_+))|$ for one $j$.
$\var_j(t'_+) = \var_j(s) + w + ( 1- \corr) - \epsilon$.
Hence, $\di(\var_j(t'_+)) = \dright(t) + (1 - \corr) + 
\epsilon - 1 = \dright(t) - \corr + \epsilon$.
We can conclude that $\dright(t'_+)=\di(\var_j(t'_+)) \leq \dih(s) + \epsilon$.
by definition of $\corr$.

\textit{Subcase 2: $\di(w)<0:$}\\
In this case, from $\dih(s) < \frac 1 4$ and $\dih(t)=\dright(t)$
it follows that $\di(var_i(s)) < 0$ for all $i$.
Thus, we set $w' = w + \ceil{w} - \epsilon$ and 
the lemma holds.

\begin{figure}[h]
\begin{center}
\begin{tikzpicture}
\small
\tikzstyle{var_style}=[draw,fill=black, circle, minimum size=1mm, inner sep=0pt];
\node[] (s) at (-2,0) {$t =$};
\draw (-1, 0) -- (5, 0);
\draw[thick] (0, 0.2) -- (0, -0.2);
\draw[] (2, 0.1) -- (2, -0.1);
\node (x1c) at (1.2,0.2){\tiny{$y_1$}};
\node[var_style] (x1) at (1.2,0){};
\node (x2c) at (0.5,0.2){\tiny{$y_2$}};
\node[var_style] (x2) at (0.5,0){};
\draw[thick] (4, 0.2) -- (4, -0.2);
\draw [snake=brace, raise snake=8pt]
    (0,0) to node [above=8pt] {$|\dleft|$} (x2);
\draw [snake=brace, raise snake=24pt]
    (0,0) to node [above=24pt] {$|\dright| > \dih(s) + \epsilon$} (x1);
\draw [snake=brace, raise snake=8pt]
    (x1) to node [above=8pt] {$1 - |\dright$|} (4,0);
\draw [snake=brace, mirror snake, raise snake=6pt]
    (x2) to node [below=8pt] {$\leq \dih(s)$} (x1);

\node[] (ss) at (-2,-3) {$t'_+ =$};
\draw (-1, -3) -- (5, -3);
\draw[thick] (0, -2.8) -- (0, -3.2);
\node (x1cp) at (3.8,-3.2){\tiny{$y_1$}};
\node[var_style] (x1p) at (3.8,-3){};
\node (x2cp) at (3.1,-3.2){\tiny{$y_2$}};
\node[var_style] (x2p) at (3.1,-3){};
\draw[] (2, -2.9) -- (2, -3.1);
\draw[thick] (4, -2.8) -- (4, -3.2);
\draw [snake=brace, raise snake=8pt]
    (x2p) to node [above=8pt] 
      {$|\dleft| \leq \dih(s) + \epsilon$} (4, -3);
\end{tikzpicture}
\end{center}
\caption{Case 2 illustration}
\label{correct2}
\end{figure}
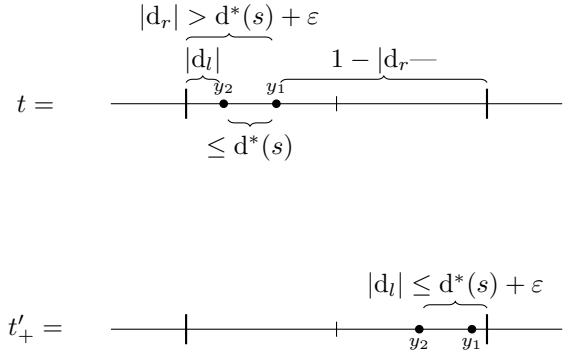

\textit{Case 3:} $\dih(t)=\dright(t)+|\dleft(t)|$.\\
We correct $w$ in the following way: $w' = w + \corr - \frac \epsilon 2$, where
$\corr = \min\{|\dleft(t)|, |d(w)|\}$.

First, we have to show that $[w'] \in [w]$ and therefore $w' \in I$.
Since $\dright(t) = \di_i(t)= \var_i(t)$ for one $i$ and 
$\dleft(t) = \di_j(t)= \var_j(t)$ for one $j$,
we can conclude from $|\di(\var_i(s)+w)| \leq |\di(\var_i(s))| + |\di(w)|$
and $|\di(\var_j(s)+w)| \leq |\di(\var_j(s))| + |\di(w)|$
and $|\di(\var_j(s)+w)|+|\di(\var_i(s)+w)| \leq |\di(\var_i(s))| + |\di(w)|+ |\di(\var_j(s))| + |\di(w)| \leq \dih(s) + 2|\di(w)|$
and $|\di(\var_j(s)+w)|+|\di(\var_i(s)+w)| > \dih(s) + \epsilon$ therefore
$|\di(w)| > \frac \epsilon 2$.
Hence, $w' \geq \floor{w}$.
Furthermore, $w' \leq \ceil{w}$, otherwise
if $\di(w) < 0$ then
assume $w' =  w + \corr - \frac \epsilon 2 > \ceil{w} = w + |\di(w)|$.
Then $\corr - \frac \epsilon 2 > |\di(w)|$. Contradiction.
Otherwise, if $\di(w) > 0$, then
$w' \leq \ceil{w}$, since by definition $\corr \leq \frac 1 2$.

Next, we have to show, that all variables that are not reset stay in the 
same interval. We consider the case, where all values of the variables are
increased, therefore we know that $\var_i(t'_+) \geq \floor{\var_i(t)}$
for all $i \not \in R$.
We now have to show that also $\var_i(t'_+) \leq \ceil{\var_i(t)}$.
Let $j$ be the index of the variable with
$\di(\var_j(t)) = \dleft(t)$.
\begin{eqnarray*}
\var_j(t'_+) & = &  \var_j(t) + w' \\ 
   & = & \var_j(t)+ w + \corr - \frac \epsilon 2\\
   & = & \var_j(t) + \corr - \frac \epsilon 2\\
   & < & \ceil{\var_i(t)} = \var_j(t) + |\dleft(t)| \\
\end{eqnarray*}

Thus we have to show: $\dih(t'_+) \leq \dih(s) + \epsilon$.
We know that $|\dright(t) - (1 + \dleft(t))| \leq \dih(s)$ and
$\dih(t'_+) = |\dleft(t'_+)| = |\di(\var_i(t'_+))|$ for $j$
such that $\di(\var_j(t))=\dright(t)$.
Also, $\var_j(t'_+) = \var_j(s) + w + \corr - \frac \epsilon 2$.
We can conclude that $\dih(t'_+) \leq \dih(s) + \frac \epsilon 2$.

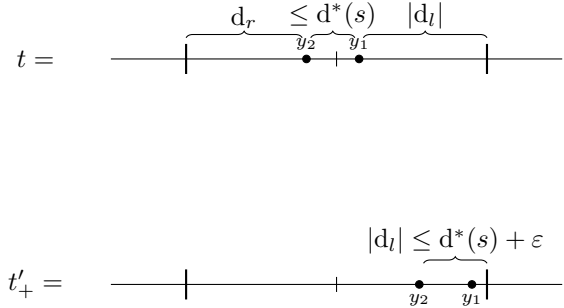
\begin{figure}[h]
\begin{center}
\begin{tikzpicture}
\small
\tikzstyle{var_style}=[draw,fill=black, circle, minimum size=1mm, inner sep=0pt];

\node[] (ss) at (-2,0) {$t =$};
\draw (-1, 0) -- (5, 0);
\draw[thick] (0, 0.2) -- (0, -0.2);
\draw[] (2, 0.1) -- (2, -0.1);
\node (x1c) at (2.3,0.2){\tiny{$y_1$}};
\node[var_style] (x1) at (2.3,0){};
\node (x2c) at (1.6,0.2){\tiny{$y_2$}};
\node[var_style] (x2) at (1.6,0){};
\draw[thick] (4, 0.2) -- (4, -0.2);
\draw [snake=brace, raise snake=8pt]
    (x1) to node [above=8pt] {$|\dleft|$} (4,0);
\draw [snake=brace, raise snake=8pt]
    (0,0) to node [above=8pt] {$\dright$} (x2);
\draw [snake=brace, raise snake=8pt]
    (x2) to node [above=8pt] {$\leq \dih(s)$} (x1);

\node[] (ss) at (-2,-3) {$t'_+ =$};
\draw (-1, -3) -- (5, -3);
\draw[thick] (0, -2.8) -- (0, -3.2);
\node (x1cp) at (3.8,-3.2){\tiny{$y_1$}};
\node[var_style] (x1p) at (3.8,-3){};
\node (x2cp) at (3.1,-3.2){\tiny{$y_2$}};
\node[var_style] (x2p) at (3.1,-3){};
\draw[] (2, -2.9) -- (2, -3.1);
\draw[thick] (4, -2.8) -- (4, -3.2);
\draw [snake=brace, raise snake=8pt]
    (x2p) to node [above=8pt] 
      {$|\dleft| \leq \dih(s) + \epsilon$} (4, -3);

\end{tikzpicture}
\end{center}
\caption{Case 3 illustration}
\label{correct3}
\end{figure}
\end{proof}

Knowing that, in one step, the move can always preserve small
total distance, we can finally define discrete strategies.

\begin{defi} \label{edisc}
We call a strategy $\sigma$ \emph{$\epsilon$-discrete}
if for every $s_{n+1} = \sigma(s_0 \ldots s_n)$ it holds that
if $\dih(s_n) \leq \epsilon$ then
$\dih(s_{n+1}) \leq \dih(s_n) + \frac{\epsilon}{2^{n+1}}$, and if for each $i$
$s'_i \sim s_i$, then $\sigma(s_0 \ldots s_n) \sim \sigma(s'_0 \ldots s'_n)$.
\end{defi}

Observe that it follows directly from the definition that if
$\dih(s_0) \leq \frac{\epsilon}{2}$ and both players play discrete
strategies, then $\dih(s_{n}) \leq \epsilon (1 - \frac 1 {2^{n+1}}).$

\begin{exa}
To see that decreasing $\epsilon$ in each step is sometimes crucial,
consider the game with one variable depicted in Figure \ref{fig-eps-need}.
In each move \pzero has to choose a positive value in $(0,1)$.
\pone can then decide to continue the play or leave the
cycle and end the play with the negative accumulated
value, \ie $-y_0$, as payoff. He cannot infinitely
often decide to stay in the cycle as then the payoff
would be $\infty$ as the priority is $0$.
An $\epsilon$-optimal strategy for \pzero as the maximising player
is thus to start with $\frac \epsilon 2$ and decrease in each step.
Please note that the value of the game is $0$.

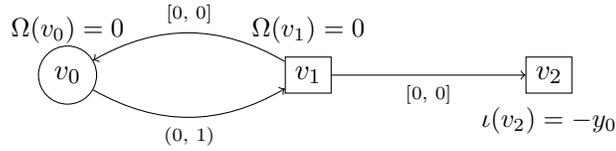
\begin{figure}[h]
\begin{center}
\begin{tikzpicture}[scale=0.8]
\node[] () at (-2, 0.75) {\small{$\Omega(v_0)=0$}};
\node[draw,circle] (a) at (-2, 0) {$v_0$};
\node[] () at (2, 0.75) {\small{$\Omega(v_1)=0$}};
\node[draw] (b) at (2, 0) {$v_1$};
\node[draw] (c) at (6, 0) {$v_2$};
\node () at (6, -0.75) {\small{$\indexi(v_2)=-y_0$}};
\path (a) edge[bend right,->] node[below] {\tiny{(0, 1)}} (b);
\path (b) edge[bend right,->] node[above] {\tiny{[0, 0]}} (a);
\path (b) edge[->] node[below] {\tiny{[0, 0]}} (c);
\end{tikzpicture}
\end{center}
\caption{Game in which the values played must decrease.}
\label{fig-eps-need}
\end{figure}
\end{exa}

We now extend the previous lemma to one that allows for the shifting of a whole move.
\begin{lem}\label{discrete_succ2}
Let $s$ be a state and $t$ a successor of $s$,
where $(s,t)$ is allowed by $l$.
Let $s'$ be a state with $\dih(s') \leq \frac 1 4$, such that $s \sim s'$. 
Then, for every $\epsilon > 0$, there exists a successor $t'$ of $s'$
allowed by $l$ such that
\begin{iteMize}{$\bullet$}
\item $s' \sim t'$ and
\item $\dih(t') \leq \dih(s') + \epsilon$.
\end{iteMize}
\end{lem}

\begin{proof}
Since $s \sim s'$ and $t \in \successor(s)$ is allowed by $l$, we know,
by Lemma \ref{shift-move}, that there exists a state $t' \in \successor(s')$
allowed by the same label $l$, such that $t' \sim t$.
We also know from Lemma \ref{discrete_succ} that,
for every choice of $\epsilon$, there exists $t_+ \in \successor(s')$
such that $\dih(t_{+}) \leq \dih(s') + \epsilon$ and $t' \sim t_+$.
Since $t' \sim t$, this also means that $t_+ \sim t$,
hence $t_+$ fulfils the requirements above.
\end{proof}

We can conclude that discrete strategies allow for the approximation of game values.

\begin{lem} \label{discrete_strat}
Fix an $\epsilon$-discrete strategy $\rho_d$ of Player $1-i$ in $\calG$,
$\epsilon < \frac{1}{4}$. For every strategy $\sigma$ of Player $i$ there
exists an $\epsilon$-discrete strategy $\sigma_d$, such that,
for every starting state $s_0$ with $\dih(s_0) < \frac{\epsilon}{2}$,
if $\pi(\sigma, \rho_d, s_0) = s_0 s_1 \ldots$ and
$\pi(\sigma_d, \rho_d, s_0) = s_0' s_1' \ldots$,
then $s_i \sim s_i'$ for all $i$.
\end{lem}

\begin{proof}
We only prove this lemma for \pzero, the case of \pone is analogous.
We define $\sigma_d$ inductively. Let $s_0$ be the starting state.
If $\sigma(s_0)=s_1$, then by Lemma~\ref{discrete_succ2} there is a
$s_1' \sim s_1$ with $\dih(s_1') \leq \dih(s_0) + \frac{\epsilon}{4}$,
and we set $\sigma_d(s_0)=s_1'$.

Let $h = s_0 \ldots s_k$ and $h' = s_0' \ldots s_k'$ be finite play
histories such that $h$ is a prefix of $\pi(\sigma, \rho_d, s_0)$
and $h'$ is consistent with $\rho_d$ and $\sigma_d$ as defined thus far.
Note that $s_0 = s_0'$ and by inductive assumption $s_i \sim s_i'$ for
$0 < i \leq k$, and $\dih(s_{k}) \leq \epsilon (1 - \frac 1 {2^{k+1}})$.
If $\sigma(s_0 \ldots s_k)=s_{k+1} \in \successor(s_k)$, then,
by Lemma~\ref{discrete_succ2}, there also exists 
a state $s_{k+1}' \in \successor(s_k')$ such that $s_{k+1}' \sim s_{k+1}$ and 
$\dih(s_{k+1}') \leq \dih(s_k) + \frac {\epsilon}{2}$.
Thus, we set $\sigma_d(s_0' \ldots s_{k}')$ to $s_{k+1}'$.
For all other histories $h'' = s''_0 \ldots s''_k$ with $s''_i \sim s_i$,
we set $\sigma(h'') = s''_{k+1}$ for any $s''_{k+1}$ equivalent with $s_k$,
which exists by Lemma~\ref{shift-move}, and we can pick a discrete one
if $\dih(s_{k}'') < \epsilon$ by Lemma~\ref{discrete_succ2}.

By construction, the strategy $\sigma_d$ is discrete and if
$\pi(\sigma, \rho_d, s_0) = s_0 s_1 \ldots$ and
$\pi(\sigma_d, \rho_d, s_0)$ $= s_0' s_1' \ldots$ then $s_i \sim s'_i$.
\end{proof}

\begin{prop} \label{values_close}
Let $\calG$ be a flat interval parity game. Let $\Gamma_i$ be the set
of all strategies for player $i$ and $\Delta_i$ the set of all discrete
strategies for player $i$ and $m$ be the highest value that occurs as
a multiplicative factor in $\indexi$. 
Then it holds, for every starting state $s$, that
\[ \left| \adjustlimits\sup_{\sigma \in \Gamma_0} \inf_{\rho \in \Gamma_1}
      \pay(\pi(\sigma, \rho,s)) - 
  \adjustlimits \sup_{\sigma \in \Delta_0} \inf_{\rho \in \Delta_1}
      \pay(\pi(\sigma, \rho,s)) \right| \leq m. \]
\end{prop}

\begin{proof}
\emph{Case 1:} assume that
\[ \adjustlimits \sup_{\sigma \in \Delta_0} \inf_{\rho \in \Delta_1} \pay(\pi(\sigma, \rho,s)) - 
\adjustlimits \sup_{\sigma \in \Gamma_0} \inf_{\rho \in \Gamma_1}
      \pay(\pi(\sigma, \rho,s)) > m. \]
Then there exists a strategy
$\sigma_d \in \Delta_0$ such that 
\[  \inf_{\rho \in \Delta_1} \pay(\pi(\sigma_d, \rho,s))
  - \inf_{\rho \in \Gamma_1} \pay(\pi(\sigma_d, \rho,s) > m.\] 
Fix a strategy $\rho_{\inf} \in \Gamma_1$, for which 
 \[\pay(\pi(\sigma_d, \rho_{\inf},s)) \leq  \inf_{\rho \in \Gamma_1} \pay(\pi(\sigma_d, \rho,s))+ \epsilon.\]
From Lemma~\ref{discrete_strat}, we know, that there is a discrete
strategy $\rho_{\inf_d} \in \Delta_1$ which is a discrete version
of $\rho_{\inf}$ against $\sigma_d$. From the above, it follows that 
$\pay(\pi(\sigma_d, \rho_{\inf_d},s))- \pay(\pi(\sigma_d, \rho_{\inf},s)) > m.$ 
This is a contradiction, since we know from Lemma~\ref{discrete_strat} that
all states in both plays are equivalent, so for finite plays also the final
states are equivalent, which means that the payoffs cannot differ by more
than $m$ as it is the highest occurring multiplicative factor in $\indexi$. If
both plays are infinite, then, by definition of $\sim$, the payoffs are equal.

\emph{Case 2:} assume that \[\adjustlimits\sup_{\sigma \in \Gamma_0} \inf_{\rho \in \Gamma_1}
      \pay(\pi(\sigma, \rho,s)) - \adjustlimits\sup_{\sigma \in \Delta_0} \inf_{\rho \in \Delta_1}
      \pay(\pi(\sigma, \rho,s)) > m.\] 
By Theorem~\ref{mccorrect} every interval parity game is determined, thus
\[ \adjustlimits \sup_{\sigma \in \Gamma_0} \inf_{\rho \in \Gamma_1}
  \pay(\pi(\sigma, \rho,s)) = 
\adjustlimits \inf_{\rho \in \Gamma_1} \sup_{\sigma \in \Gamma_0}
    \pay(\pi(\sigma, \rho,s)). \]
In the next section, we show that restricting to discrete strategies
corresponds to playing a counter-reset game, and since these are again
determined games, we get that
\[ \adjustlimits \sup_{\sigma \in \Delta_0} \inf_{\rho \in \Delta_1}
  \pay(\pi(\sigma, \rho,s)) = 
\adjustlimits \inf_{\rho \in \Delta_1} \sup_{\sigma \in \Delta_0}
    \pay(\pi(\sigma, \rho,s)). \]
Therefore we can rewrite the assumption of this case as
\[\adjustlimits \inf_{\rho \in \Gamma_1} \sup_{\sigma \in \Gamma_0}
    \pay(\pi(\sigma, \rho,s)) - 
    \adjustlimits \inf_{\rho \in \Delta_1} \sup_{\sigma \in \Delta_0}
      \pay(\pi(\sigma, \rho,s)) > m.\]
Then there exists a strategy $\rho_d \in \Delta_1$ such that 
\[ \sup_{\sigma \in \Gamma_0} \pay(\pi(\sigma, \rho_d, s)) -
   \sup_{\sigma \in \Delta_0} \pay(\pi(\sigma, \rho_d, s)) > m. \] 
Fix a strategy $\sigma_{\sup} \in \Gamma_0$, for which 
\[ \pay(\pi(\sigma_{\sup}, \rho_d, s)) \geq
   \sup_{\sigma \in \Gamma_0} \pay(\pi(\sigma, \rho_d, s)) - \epsilon.\]
From Lemma~\ref{discrete_strat}, we know, that there again is a discrete
strategy $\sigma_{\sup_d} \in \Delta_0$ which is a discrete version
of $\sigma_{\sup}$ against $\rho_d$. From the above, it follows that 
$\pay(\pi(\sigma_{\sup}, \rho_d, s))- \pay(\pi(\sigma_{\sup_d}, \rho_d,s)) > m,$
which again contradicts that all states in these two plays are equivalent.
\end{proof}

\section{Counter-Reset Games} \label{sec_pushdown}

In this section, we introduce counter-reset games and show,
using the discretisation results from the previous section, 
that approximating the value of an IPG over $\Zinf$ can be
reduced to solving a counter-parity game. We then solve these
games using an algorithm from \cite{BKL12}.

By Proposition~\ref{values_close} above, we can restrict both players
in a flat IPG to use $\epsilon$-discrete strategies to approximate the value
of a flat interval game up to the maximal multiplicative factor $m$.
Multiplying the game by any number $q$ does not change the
multiplicative factors in $\iota$ but multiplies the value
of the game by $q$. Thus, to approximate the value of $\calG$
up to $\frac 1 n$ it suffices to play $\epsilon$-discrete
strategies in $n \cdot m \cdot \calG$. When players use only
discrete strategies, the chosen values remain close to integers
(possibly being up to $\epsilon$ bigger or smaller).
Whether the value is bigger, equal or smaller than an integer
can be stored in the state, as well as whether the value of
a variable is smaller or bigger than any of the (non-infinite) bounds
in constraint intervals. This way, we can eliminate both $\epsilon$'s
and constraints and are left with the following games.

\begin{defi}
A \emph{counter-reset} game is a flat interval parity game in which
in each label $l = (I, \ol{C}, R)$ the constraints $\ol{C}$ are
trivially true and the interval $I$ is either $[0, 0]$ or $[1, 1]$,
\ie either all variables are incremented by $1$ or all are left intact.
\end{defi}

\begin{exa}
In Figure \ref{counterex}, we depict a simple counter-reset game.
As usual, circles represent positions of \pzero and boxes those of
\pone. Priorities, payoff functions, intervals and reset sets are
also depicted as usual next to the corresponding nodes or above
transitions.
In this game, we have two variables, $y_0, y_1$ and as mentioned
above, there are no constraints on these variables in counter-reset
games, but they can be reset. The only choice in this game that \pzero
has is to increase all variables (``choose'' 1 from $[1,1]$) 
and \pone can do the same
or end the game and get a payoff of $-y_0$. Since he wants to minimise, his
best strategy is to loop as long as possible but not infinitely long, as 
the lowest priority on the according cycle is 0. 
Since he can achieve arbitrary small values this way, the value of this game 
(starting at $v_0$ or $v_1$) is $-\infty$.
\end{exa}

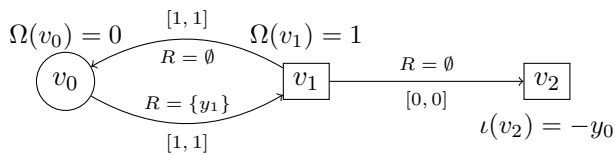
\begin{figure}[h]\label{counterex}
\begin{center}
\begin{tikzpicture}[scale=0.8]
\node[] () at (-2, 0.75) {\small{$\Omega(v_0)=0$}};
\node[draw,circle] (a) at (-2, 0) {$v_0$};
\node[] () at (2, 0.75) {\small{$\Omega(v_1)=1$}};
\node[draw] (b) at (2, 0) {$v_1$};
\node[draw] (c) at (6, 0) {{$v_2$}};
\node () at (6, -0.75) {\small{$\indexi(v_2)=-y_0$}};
\path (a) edge[bend right,->] node[below] {{\tiny{$[1, 1]$}}} node[above]
{{\tiny{$R=\{y_1\}$}}} (b);
\path (b) edge[bend right,->] node[above] {{\tiny{$[1, 1]$}}} node[below]
{{\tiny{$R=\emptyset$}}} (a);
\path (b) edge[->] node[below] {{\tiny{$[0, 0]$}}} node[above]
{{\tiny{$R=\emptyset$}}}(c);
\end{tikzpicture}
\end{center}
\caption{Simple counter-reset game}
\end{figure}

\begin{lem} \label{cr_approx}
Let $\calG$ be an IPG over $\Zinf$ with maximal absolute value of 
the multiplicative factor in $\iota$ equal to $m$. For each $n \in \bbN$
there exists a counter-reset game $\calG'_n$ such that for all states $s$
in which all variables are integers:
\[ \left| \val\calG(s) - \frac{\val\calG'_n(n \cdot m \cdot s)}
     {n \cdot m} \right| \leq \frac 1 n. \]
\end{lem}

\begin{proof}
Consider first the game $\calG'' = n \cdot m \cdot \calG$.
By construction, the multiplicative factors in $\iota$ do not
change and thus their maximal value in $\calG''$ is still $m$.
By Lemma \ref{mult_games}, in all states $s$ holds
\[ \val\calG(s) = \frac{ \val\calG''(s) }{n \cdot m}. \]
Moreover, by Proposition \ref{values_close} applied to $\calG''$
\[ \left| \val\calG''(s) - \adjustlimits \sup_{\sigma \in \Delta_0}
     \inf_{\rho \in \Delta_1} \pay(\pi_{G''}(\sigma, \rho,s)) \right| \leq m, \]
and therefore
\[ \left| \val\calG(s) - \frac{\adjustlimits \sup_{\sigma \in \Delta_0}
     \inf_{\rho \in \Delta_1} \pay(\pi_{G''}(\sigma, \rho,s))}{n \cdot m} \right|
   \leq \frac 1 n. \]

We will now show how to construct the counter-reset game $\calG'$ with value equal to
$\adjustlimits \sup_{\sigma \in \Delta_0} \inf_{\rho \in \Delta_1} \pay(\pi_{G''}(\sigma, \rho,s))$,
i.e. to the value of $\calG''$ when both players play $\epsilon$-discrete strategies.
To this end, we first construct the game $\calG'_0$ which still has constraints, but
in which all intervals are $[k,k]$ for some $k \in \bbN$. The game $\calG'_0$ is
constructed from $\calG''$ by replacing each position $v$ by $3^M$ positions $v^{i_1 \ldots i_M}$.
The sequence $i_1 \ldots i_M \in \{-1, 0, 1\}^M$ keeps track, for each variable, whether
it is currently smaller, greater, or equal to an integer. The interval labels are now
converted in the following way. If a move with interval $[n, n+k)$ and resets $R$ is taken
from a position $v^{i_1 \ldots i_M}$ in $\calG'_0$ and would lead to $w$ in $\calG''$, then
a sequence of moves with labels $[l, l]$ for each $n \leq l \leq n+k$ is added,
with the $l$-labelled move leading to $w^{j_1 \ldots j_M}$ such that:
\begin{iteMize}{$\bullet$}
\item if one $j_k > i_k$ then all $j_k > i_k$ for $k \in \{0,\ldots,M\}$, and the same if $j_k<i_k$ or $j_k=i_k$,
\item if $l = n$ then each $j_k \geq i_k$ (interval was downwards-closed), and
\item if $l = k$ then each $j_k < i_k$ (interval was upwards-open). 
\end{iteMize}
The situation for open, closed, and open-closed intervals is analogous.
The plays which use discrete strategies in $\calG''$ can now be directly
transferred to plays in $\calG'_0$ in which indeed in $v^{i_1\ldots i_M}$
the sign of the fractional part of $y_j$ is equal to $i_j$. The same
can be done in the other direction, as the constraints listed above
allow to choose a value in the interval which leads to the appropriate
change in the sign sequence. Therefore 
\[ \val\calG'_0 = \adjustlimits \sup_{\sigma \in \Delta_0} \inf_{\rho \in \Delta_1} \pay(\pi_{G''}(\sigma, \rho,s)). \]

To eliminate the constraints from move labels in $\calG'_0$ we
determine the highest non-infinite bound $b$ which appears in these
constraints (both on the left and on the right side of an interval).
Then, we construct $\calG'$ as the synchronous product of $\calG'_0$
with a memory of size $(b+2)^M$ which remembers, for each variable $y_i$,
whether $y_i$ is greater than $b$ or equal to $b, b-1, \ldots, 0$. With
this memory, we resolve all constraints and remove them from move labels
in $\calG'$.
\end{proof}

Counter reset games are another representation of a class of counter parity
games, which were recently studied in \cite{BKL12}, where an algorithm to solve
such games was given, improving our previous decidability result \cite{FK11}.

\begin{thm}[\cite{BKL12}]
For any finite counter parity game $\calG$ and initial vertex $v$,
the value $\val \calG(v)$ can be computed in \textsc{6EXPTIME}.
When the number of counters is fixed, the value can be computed
in \textsc{4EXPTIME}.
\end{thm}

\begin{cor} \label{cr_value}
For any finite counter reset game $\calG$ with a starting state $s$ where all
counters are integers, the value $\val \calG(s)$ can be computed in
\textsc{6EXPTIME}. With fixed number of counters, the value
can be computed in \textsc{4EXPTIME}.
\end{cor}

\section{Conclusions and Future Work}

We conclude by completing the proof of our main Theorem~\ref{mainthm}.
We first observe that, by Theorem~\ref{mccorrect}, evaluating a
\Qmu-formula on a system is equivalent to calculating the value of the 
corresponding model-checking game. We can then turn this game into a
flat one by Lemma~\ref{flat} and then into one over $\Zinf$ by
Corollary~\ref{mult_games_col}.
By Lemma~\ref{cr_approx} the value of such a game can be approximated
with arbitrary precision by counter-reset games, which we can solve
by Corollary~\ref{cr_value}.

All together, we proved that it is possible to approximate the values
of quantitative $\mu$-calculus formulae on initialised linear hybrid
systems with arbitrary precision. With the recent result on counter parity
games, we are even able to provide an elementary algorithm -- as the game
$\calG'_n$ in Lemma~\ref{cr_approx} is doubly-exponential in $\calG$ and $n$,
the combined complexity of the above procedure is \textsc{8EXPTIME}
(note the doubly-exponential increase compared to Corollary~\ref{cr_value}).

This complexity is very high and the complexity bound is not tight,
thus we can formulate two immediate open problems: (1) can the exact
value of $\semk{\phi}$ be computed? (2) what is the exact complexity
of such a computation or its approximation? 
Another open question is whether we can use our methods for more general
classes of games, e.g. for games with more complex payoff functions such as
mean-payoff interval games. Furthermore, we are not only interested in
theoretical complexity bounds but also in the practical applicability of
quantitative model checking. This will require a more thorough algorithmic
analysis of the problem. Also, since we reduce the problem to counter parity
games, the implementation of a solver for this class of games is a necessary
first step before we can exploit the methods presented in this paper in
practice. However, even with further research needed to answer these challenges,
our result lays the foundation for using quantitative temporal logics in
the verification of hybrid systems.

\end{document}